\documentclass[%
    reprint,
    superscriptaddress,
    longbibliography,
    nofootinbib,
    amsmath,
    amssymb,
    aps,
    floatfix,
]{revtex4-2}
\usepackage{graphicx}
\usepackage{dcolumn}
\usepackage{bm}
\usepackage[utf8]{inputenc}
\usepackage{acronym}

\usepackage{svg}
\usepackage{dsfont}
\usepackage{amsmath,amsfonts,amssymb,amsthm}
\usepackage{braket}
\usepackage{physics}
\usepackage{tikz}
\usetikzlibrary{quantikz}
\usepackage{acronym}
\usepackage{xspace}
\usepackage{siunitx}
\usepackage{booktabs}
\usepackage[hidelinks]{hyperref}
\usepackage{nicefrac}       
\usepackage{microtype}      
\usepackage{mathtools}
\usepackage{youngtab}
\usepackage{comment}

\usepackage{listings}
\usepackage{todonotes} 
\usepackage{siunitx}
\usepackage{multirow}
\usepackage{tabularx}

\newtheorem{claim}{Claim}
\newtheorem{theorem}{Theorem}
\newtheorem{definition}{Definition}
\newtheorem{proposition}{Proposition}

\newtheorem{lemma}{Lemma}
\newtheorem{example}{Example}

\newcommand{\irrep}{{\rho_\lambda}}
\newcommand{\qhat}{\hat{q}_{\lambda}}
\newcommand{\hhat}{\hat{h}_{\lambda}}
\newcommand{\sn}{\mathbb{S}_n}

\begin{document}

\title{Probabilistic modeling over permutations using quantum computers}
\author{Vasilis Belis$^{*\dagger}$}
\affiliation{Xanadu Quantum Technologies, Toronto, ON, M5G 2C8, Canada}
\thanks{Corresponding author: \href{mailto:vasilis.belis@xanadu.ai}{vasilis.belis@xanadu.ai}}
\author{Giulio Crognaletti$^{\dagger}$}
\affiliation{Department of Physics, University of Trieste, Strada Costiera 11, 34151 Trieste, Italy}
\affiliation{European Organization for Nuclear Research (CERN), CH-1211 Geneva, Switzerland}
\author{Matteo Argenton$^{\dagger}$}
\affiliation{Istituto Nazionale di Fisica Nucleare, Sezione di Ferrara
Via Saragat 1, Ferrara, Italy}
\affiliation{Department of Physics and Earth Science, University of Ferrara, Via Saragat 1, Ferrara, 44122, Italy}
\altaffiliation{Contributed equally.}
\affiliation{European Organization for Nuclear Research (CERN), CH-1211 Geneva, Switzerland}
\author{Michele Grossi} 
\affiliation{European Organization for Nuclear Research (CERN), CH-1211 Geneva, Switzerland}
\author{Maria Schuld}  
\affiliation{Xanadu Quantum Technologies, Toronto, ON, M5G 2C8, Canada}

\date{\today}
\begin{abstract}
    Quantum computers provide a super-exponential speedup for performing a Fourier transform over the symmetric group, an ability for which practical use cases have remained elusive so far.  In this work, we leverage this ability to unlock spectral methods for machine learning over permutation-structured data, which appear in applications such as multi-object tracking and recommendation systems. It has been shown previously that a powerful way of building probabilistic models over permutations is to use the framework of non-Abelian harmonic analysis, as the model's group Fourier spectrum captures the interaction complexity: ``low frequencies'' correspond to low order correlations, and ``high frequencies'' to more complex ones. This can be used to construct a Markov chain model driven by alternating steps of \textit{diffusion} (a group-equivariant convolution) and \textit{conditioning} (a Bayesian update). However, this approach is computationally challenging and hence limited to simple approximations. Here we construct a quantum algorithm that encodes the exact probabilistic model---a classically intractable object---into the amplitudes of a quantum state by making use of the Quantum Fourier Transform (QFT) over the symmetric group. We discuss the scaling, limitations, and practical use of such an approach, which we envision to be a first step towards useful applications of non-Abelian QFTs. 
\end{abstract}

\maketitle

\section{Introduction}\label{sec:intro}
Probabilistic modeling over permutations presents a significant computational challenge central to problems such as identity management in multi-object tracking and preference learning in recommender systems. Mathematically, these problems can be naturally formulated within the framework of non-commutative harmonic analysis, where the statistical structure of the data is intimately linked to the properties of the symmetric group $\sn$~\cite{diaconis1988group}---the group containing all possible orderings of $n$ objects (refer to ~App.~\ref{app:representation_theory} and~App.~\ref{app:irrep_sn} for a primer on group representation theory). The power of the harmonic analysis framework lies in the interpretation of the \textit{Fourier spectrum} of the model, which enables the design of models with favorable learning properties. Here, a ``frequency'' corresponds to the complexity of interactions: low-frequency components encode single-object statistics (first-order marginals) and simple pairwise relationships, while high-frequency components encode intricate multi-body dependencies~\cite{diaconis1988group}. 
The group Fourier-theoretic angle therefore helps to build probabilistic models, as pioneered by Risi Kondor~\cite{pmlr-v2-kondor07a, kondor2007skew, kondor2010fourier, kondor2012multiresolution} and Jonathan Huang~\cite{JMLR:v10:huang09a, huang2007efficient}. 

However, the computational complexity of working with permutations forces truncating the spectrum and retaining only the low-frequency components---a technique known as \textit{band-limiting}---which means that the model only captures lower-order correlations. Even with these approximations, the resource requirements remain severe: to model a probability distribution up to $k$-th order marginals in Fourier space---i.e., capturing correlations among subsets of $k$ objects---requires classical memory and compute time scaling as $\mathcal{O}(n^{2k})$ (cf.~App.~\ref{app:spectrum_stat_interp}). In addition to these limitations, model building and inference has to be performed in Fourier space, as moving between Fourier and direct space with the Fast Fourier Transform (FFT) over $\sn$ would scale with $\mathcal{O}(n!\,n^2)$, and hence super-exponentially~\cite{maslen:1998}. 

\begin{figure*}[t]
    \includegraphics[width=0.8\textwidth]{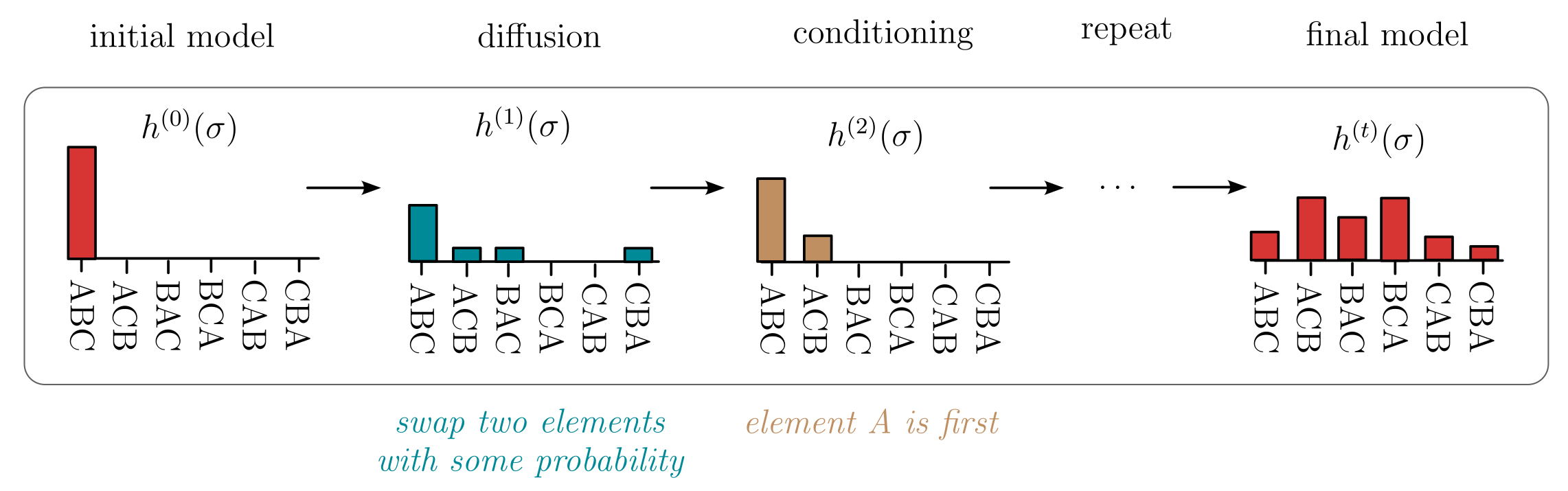}
       \caption{ \textbf{The probabilistic modeling process:} The evolution of the belief state over permutations of $n$ objects is captured by the probability distribution $h^{(t)}(\sigma)$. Initialization ($t=0$) begins with a deterministic assignment (canonical configuration), visualized here for $n=3$ as $ABC \to[1,2,3]$. The process alternates between two operations: \textit{diffusion}, which spreads probability mass (increasing entropy) to model the uncertainty in the system between steps; and \textit{conditioning}, a Bayesian update that refines the distribution based on partial observations that are the training data. The process repeats for each observation in the training data.}
    \label{fig:markov}
\end{figure*}
\begin{figure*}[t]
    \includegraphics[width=0.8\textwidth]{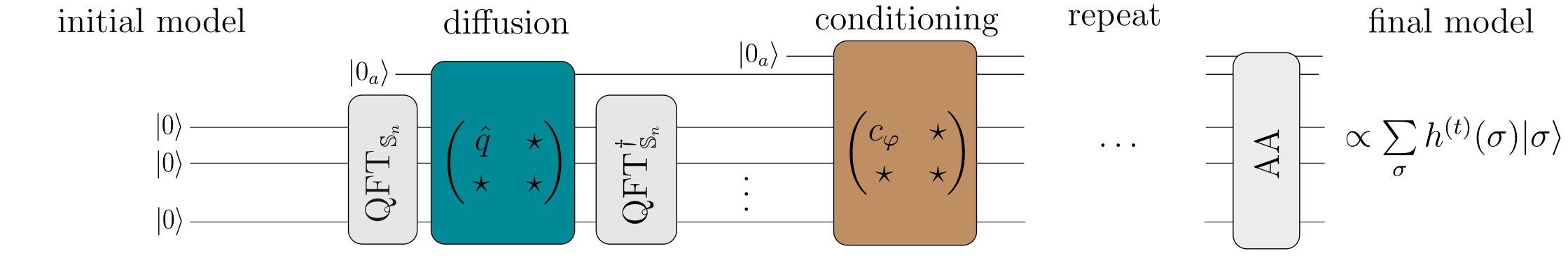}
   \caption{\textbf{Our quantum implementation}: The probability distribution is encoded in the amplitudes of the quantum state $\ket{\psi^{(t)}}\propto\sum_\sigma h^{(t)}(\sigma)\ket{\sigma}$. We use Lehmer's code to map permutations $\sigma\in\sn$ to $n!$ integers, which are then encoded to binary represented by computational basis states $\ket{\sigma}$.
   We implement diffusion in Fourier space by block-encoding an operator $\hat{q}$ (see Sec.~\ref{sec:diffusion:block_encoding}). The conditioning step is implemented by a block-encoding of the operator $c_{\varphi}$ in direct space that exploits Lehmer's code. The Quantum Fourier Transform (QFT) over $\sn$ efficiently toggles between Fourier and direct space. Each of these block-encodings requires a new ancilla. At the end, Amplitude Amplification boosts the probability that the output state encodes the full posterior probability distribution over permutations.}
    \label{fig:qalgo}
\end{figure*}

These two issues, limiting the model's capacity and working only in Fourier space, are elegantly solved on a quantum computer. A quantum state can store and manipulate both the $n!$-dimensional probabilistic model and its Fourier spectrum, and one can move between the two representations efficiently: the Quantum Fourier Transform (QFT) over the symmetric group~\cite{Beals:1997,Moore:2003icv,Kawano:2016} has a runtime of only $\mathcal{O}(n^3\log n)$~\cite{Kawano:2016}. In this work, we aim to understand if this remarkable ability of quantum computers can unlock the group-theoretical framework for probabilistic modeling over permutations. 
To this end, we follow a particular recurrent (i.e., Markovian) learning framework suggested by Kondor and Huang~\cite{JMLR:v10:huang09a, pmlr-v2-kondor07a}, and formulate an equivalent quantum algorithm that, under certain assumptions that we argue are realistic in learning problems, prepares a quantum state encoding a probabilistic model over $\sn$, but without any need for strong band-limiting and approximations. 

Specifically, we encode permutations in computational basis states using $\mathcal{O}(n\log n)$ qubits and propose quantum implementations for \textit{diffusion} and \textit{conditioning}---the two critical operations that drive the evolution of the belief state over object configurations in this Markovian setting (Fig.~\ref{fig:markov}). Using tools such as block-encoding and amplitude amplification (Fig.~\ref{fig:qalgo}) we show that, at least for a small number of diffusion and conditioning steps, and under reasonable assumption with regards to the ground truth distribution, the quantum algorithm runs in time dominated by the runtime of the QFT $\mathcal{O}(n^3\log n)$. While further investigations are needed, in principle this has the potential to yield a super-exponential speedup over classical exact methods for a \emph{useful} task. Beyond algorithm design, we investigate the conditions under which our method remains scalable with $n$ (Sec.~\ref{sec:model}). We discuss how the resulting quantum states can be utilized for downstream tasks (Sec.~\ref{sec:model_use}), such as maximum-a-posteriori estimates, and generative modeling---depending on the specific data encoding strategy---and pinpoint remaining open questions (Sec.~\ref{sec:discussion}). 

We emphasize that our proposal is a first step, or ``feasibility study'', towards practically useful quantum machine learning with spectral methods for $\sn$. Moving forward, numerical experiments will be needed to provide compelling evidence that the resulting machine learning model yields strong performance on real-world datasets, and that data-dependent assumptions are met both in practice and in near-term fault-tolerant quantum computers. Such simulations, however, require efficient classical implementations of the Fast Fourier Transform over $\sn$, which exhibits super-exponential scaling, as well as thoughtfully designed benchmarks to generate meaningful empirical data from necessarily small-scale settings. 

Ultimately, we envision the methods developed here will inspire further algorithmic primitives that harness the efficiency of the $\sn$-QFT. While this transform represents a capability of quantum computers that has seen renewed theoretical interest~\cite{Larocca:2024waq,Bravyi:2023veg}, its connection to practical tasks has so far remained elusive~\cite{shehab2017quantum, Childs:2008shx}.

\section{Fundamentals}
Before introducing our quantum algorithm shown in Figure~\ref{fig:qalgo}, we start off with some background on the (Quantum) Fourier Transform over the symmetric group. While the matrix-valued Fourier coefficients can look intimidating at first, the intuition is analogous to standard harmonic analysis, which can be understood as Fourier transforms of functions over Abelian groups like $(\mathbb{R}, +)$ or $\mathbb{Z}_n$.

\subsection{Harmonic Analysis on the Symmetric Group}
\label{sec:harmonic_analysis}

A function defined on a finite group $h:G\to \mathbb{C}$ can be analyzed using the Group Fourier Transform (GFT). The Fourier coefficients are matrix-valued and given by:
\begin{equation}
    \hat{h}_\rho = \sum_{x\in G} h(x)\rho(x),
    \label{eq:group_fourier_transform}
\end{equation}
where $\rho$ is an irreducible representation (irrep) of $G$ with dimension $d_\rho$. 
The original function is recovered via the \textit{inverse Fourier transform}:
\begin{equation}
    h(x) = \frac{1}{|G|} \sum_{\rho\in\hat{G}} d_\rho \Tr\left(\rho(x)^\dagger \cdot \hat{h}_{\rho}\right),
    \label{eq:inverse_group_fourier_transform}
\end{equation}
where $\hat{G}$ denotes the set of all irreps (up to isomorphism) and we use the unitary property $\rho(x)^\dagger = \rho(x^{-1})$. Unlike classical harmonic analysis on Abelian groups—such as the integers or reals where coefficients are scalars indexed by frequencies—here the coefficients are matrices indexed by representations. In this work we focus on probabilistic modeling over the symmetric group $G=\sn$ (of order $n!$). The irreps of $\sn$ are indexed by partitions of $n$, i.e., tuples of a non-increasing sequence $\lambda = (\lambda_1,\lambda_2,\dots,\lambda_\ell)$ where $|\lambda|\coloneq\sum_i \lambda_i=n$, which are denoted $\lambda \vdash n$. We write the Fourier coefficients as $\hhat$ and assume without loss of generality that the irreps $\rho_\lambda$ are unitary.
For more details on group theory and the representation theory of $\sn$ see App.~\ref{app:representation_theory}.

\paragraph*{Statistical Interpretation.}
Beyond its utility as an algebraic tool, the GFT offers rich statistical insights, a perspective pioneered by Diaconis for the symmetric group~\cite{diaconis1988group}. If $h$ is a probability distribution, the Fourier coefficient $\hhat$ is interpreted as the \textit{expectation value} of the representation matrix: $\mathbb{E}_{\sigma\sim h}[\rho_\lambda(\sigma)]$. Thus, $\hhat$ measures how much of the structure encoded by the specific symmetry $\irrep$ is preserved in the distribution. We explain in App.~\ref{app:spectrum_stat_interp} why this relates to expected patterns such as \textit{item $A$ is in position $3$} (which is a first-order frequency/marginal), or \textit{items $A, C, D$ are in positions $1,3,5$ while $B, E$ are in $2,4$} (which is a higher-order one).

\subsection{The Quantum Fourier Transform}
\label{sec:qft_operator}
A function $h:\sn\to\mathbb{C}$ can be represented as a vector $\ket{h}\in\mathcal{H}$ in a suitably large Hilbert space, and expanded as a superposition state in the computational basis:
\begin{equation}
    \ket{h} = \frac{1}{\sqrt{|\sn|}}\sum_{\sigma\in \sn} h(\sigma)\ket{\sigma}.
\end{equation}
The computational basis $\{\ket{\sigma}\}$ is said to span the \textit{regular representation} of the group. The group action on this space is represented by unitary matrices $R:G\to\mathbb{U}(\mathcal{H})$, which permute basis elements via left-multiplication $R(\tau)\ket{\sigma} = \ket{\tau\sigma}$. 

Because the regular representation is \textit{reducible}, the Hilbert space $\mathcal{H}$ decomposes into a direct sum of invariant subspaces ($G$-modules). A key property of the regular representation is that it contains \textit{every} irrep $\rho_\lambda \in \hat{G}$, appearing with multiplicity equal to its dimension $d_\lambda$. We can explicitly write this decomposition using a tensor product of the representation space $V_\lambda$ and a multiplicity space $\mathbb{C}^{d_\lambda}$:
\begin{equation}
    \mathcal{H} \cong \bigoplus_{\lambda\vdash n} (V_\lambda \otimes \mathbb{C}^{d_\lambda}).
\end{equation}
In this basis, the group action $R(\sigma)$ is block-diagonalized:
\begin{equation} 
    R(\sigma) \cong \bigoplus_{\lambda\vdash n} \left( \rho_\lambda(\sigma) \otimes \mathbb{I}_{d_\lambda} \right). 
\end{equation}
The \textit{Quantum Fourier Transform} is the unitary change-of-basis operator $\mathcal{F}$ that maps the computational basis to this Fourier basis $\ket{\lambda, i, j}$:
\begin{equation}
    \mathcal{F} = \sum_{\sigma\in \sn}\sum_{\lambda\vdash n}\sum_{i,j=1}^{d_\lambda}\sqrt{\frac{d_\lambda}{n!}}\left[\rho_\lambda(\sigma)\right]_{ij}\ket{\lambda ij}\bra{\sigma},
\end{equation}
where the basis states are grouped in a tuple that indexes the irreps $\lambda$ and the dimensions of the corresponding space $V_\lambda$ ($G$-module) indexed by $i$ and $j$.

\paragraph*{Permutation encoding.} To encode the $n!$ possible permutations into computational basis states, we employ the \textit{Lehmer code}, which maps each permutation to a unique integer index (see below). This encoding requires a register of size $\lceil \log_2(n!) \rceil \in \mathcal{O}(n\log n)$ qubits. Crucially, we adopt this specific scheme because it aligns with the factorial number system structure required to implement the efficient QFT over the symmetric group~\cite{Beals:1997,Moore:2003icv,Kawano:2016}.

\paragraph*{QFT Implementation.} Analogous to the construction of the standard QFT over $\mathbb{Z}_N$, efficient quantum circuits for the QFT over $\sn$ utilize the recursive structure of generalized FFT algorithms~\cite{Beals:1997,Moore:2003icv}.
This yields an \textit{exponential quantum speedup}: the classical FFT over the symmetric group scales with $\mathcal{O}(n!\,n^2)$~\cite{maslen:1998}, while the most efficient quantum implementation known to us has a runtime of $\mathcal{O}(n^3\log n)$~\cite{Kawano:2016}.

\section{Constructing the quantum model}\label{sec:model}

In this section, we detail the construction of the quantum model. 
The model can be viewed as a Markov-chain driven by two alternating operations: 
\begin{enumerate} 
\item \textit{Diffusion}: Models the growth of uncertainty in the belief state between conditioning steps. This is implemented as a group-equivariant convolution, akin to operations in Geometric Deep Learning~\cite{bronstein2021geometric}, equivalent to a random walk on the Cayley graph of $\sn$ generated by transpositions---permutations that swap two elements. 
\item \textit{Conditioning}: Updates the belief state via Bayes' rule upon receiving new classical data. This acts as a symmetry-breaking operation---breaking the group equivariance of diffusion---which concentrates probability mass based on new observations (e.g., a sensor detecting object tracks-locations, or a user submitting a new preference ranking). Note that ``observations'' here refer to these classical data inputs that are sequentially loaded, \textit{not} quantum measurements of the state.
\end{enumerate}

The core subroutine enabling the computational efficiency of our algorithm is the Quantum Fourier Transform (QFT) over $\sn$. In particular, it allows us to exploit a fundamental duality: diffusion acts diagonally in Fourier space, while conditioning acts diagonally in direct space. The QFT enables efficient toggling between these two representations.

\paragraph*{Classical limitations.} Classically, exact inference is bottlenecked not only by memory requirements but also by the prohibitive runtime of the Fast Fourier Transform (FFT) required to alternate between these bases~\cite{pmlr-v2-kondor07a,huang2007efficient}. 
Attempting to bypass this toggling bottleneck by performing both diffusion and conditioning entirely in Fourier space is equally intractable. Conditioning in Fourier space requires computing Kronecker (Clebsch-Gordan) coefficients for the symmetric group, a problem\footnote{Computing Kronecker coefficients as well as other representation-theoretic multiplicities is hard in general also for quantum computers~\cite{Bravyi:2023veg,Larocca:2024waq,Ikenmeyer:2023cgl,Christandl:2026clo}.} known to be \#P-hard~\cite{fischer2020computationalcomplexityplethysmcoefficients,Ikenmeyer_2017}---and even determining their positivity is NP-hard~\cite{Ikenmeyer_2017}.
To cope with these limitations, practitioners are forced to resort to coarse, band-limited approximations in Fourier space, sacrificing exactness~\cite{pmlr-v2-kondor07a,huang2007efficient}.

By leveraging the QFT, our algorithm circumvents these classical limitations. Under realistic assumptions, detailed in Sec.~\ref{sec:diffusion} and~\ref{sec:conditioning}, this approach renders the \textit{exact} Markov chain dynamics computationally tractable.

Following the construction of the model, we discuss how the resulting quantum state can be utilized in practice.
Our approach deviates from the options outlined in~\cite{kondor2011non, JMLR:v10:huang09a}, whose inference tasks rely on, and is restricted to, the computation of low-frequency Fourier coefficients. We outline methods for using the state as a generative model for sampling, as an input for downstream classical heuristics, or for tasks requiring the identification of the most likely permutation. 

\subsection{Quantum Modeling of Distributions}
\label{sec:modelling_setup}

Let $h^{(t)}(\sigma)$ denote the probability distribution representing the belief state over permutations at time $t$. We define the corresponding \textit{quantum model} as the state:
\begin{equation}
    \ket{\psi^{(t)}} = \frac{1}{\sqrt{N^{(t)}}}\sum_{\sigma \in \sn} h^{(t)}(\sigma)\ket{\sigma},
    \label{eq:def_quantum_model}
\end{equation}
where $N^{(t)}=\norm{h^{(t)}}_2^2\coloneqq\sum_\sigma |h^{(t)}(\sigma)|^2$ ensures $\ell^2$-normalization of the state.
Applying the QFT operator $\mathcal{F}$ yields the model in Fourier space:
\begin{equation}
    \ket{\hat{\psi}^{(t)}} \coloneqq \mathcal{F}\ket{\psi^{(t)}} 
    = \sum_{\lambda\vdash n}\sum_{i,j=1}^{d_\lambda} \left[\hat{h}^{(t)}_{\lambda}\right]_{ij}\ket{\lambda ij},
    \label{eq:quantum_model_fourier_time_t}
\end{equation}
where the coefficients are given by the unitary-normalized GFT: $\hat{h}^{(t)}_{\lambda} \coloneqq \sqrt{\frac{d_\lambda}{n!}} \sum_\sigma h^{(t)}(\sigma)\rho_\lambda(\sigma)$.

\paragraph*{Initial State.}
One natural choice for an initial state ($t=0$) is a delta-peak distribution concentrated on the identity element $e$, $h^{(0)}(\sigma) = \delta_{\sigma e}$, representing the canonical initial assignment of objects to identities for multi-object tracking or ranking tasks.
In the computational basis, this is simply:
\begin{equation}
    \ket{\psi^{(0)}} = \ket{e} \equiv \ket{0}.
\end{equation}
In Fourier space, this state becomes a superposition\footnote{Akin to the intuition from abelian Fourier analysis and the uncertainty principle, where a localization in one space corresponds to a ``spreading'' in dual space.} of maximally entangled states in each irrep block:
\begin{align}
    \ket{\hat{\psi}^{(0)}} &= \sum_{\lambda\vdash n}\sum_{i,j=1}^{d_\lambda}\sqrt{\frac{d_\lambda}{n!}} \underbrace{\left[\rho_\lambda(e)\right]_{ij}}_{\delta_{ij}}\ket{\lambda ij} \\
    &= \sum_{\lambda\vdash n}\sqrt{\frac{d_\lambda}{n!}} \sum_{i=1}^{d_\lambda} \ket{\lambda ii}.
    \label{eq:init_state}
\end{align}
Note that in Sec.~\ref{sec:model_use}, we present an alternative initial state and show that the main operations of our algorithm (diffusion, conditioning, and QFT) remain the same.

\subsection{Diffusion Operation}\label{sec:diffusion}
We consider a diffusion process induced by a random walk on the Cayley graph of $\sn$ generated by transpositions. 
In this context, the configuration $\sigma^{(t+1)}$ is generated from $\sigma^{(t)}$ by applying a random permutation $\tau^{(t)}$ drawn from a probability distribution $q$: 
\begin{equation}
    \sigma^{(t+1)} = \tau^{(t)}\sigma^{(t)}.
\end{equation}
We call $q$ the \textit{diffusion kernel}.
For a single diffusion step, the time evolution of the probability distribution that models our system is given by convolving the diffusion kernel $q$ with the probability distribution (see derivation in App.~\ref{app:diffusion}):
\begin{align}
    h^{(t+1)}(\sigma^{})&=\left(q\,\star\,h^{(t)}\right)(\sigma)\\
    & \coloneqq \sum_{\tau\in\sn} q(\sigma \,\tau^{-1})\,h^{(t)}(\tau)
\end{align}
Analogously to the abelian case, the convolution operation is a point-wise product in Fourier space; here, a matrix multiplication (see Convolution theorem in App.~\ref{app:representation_theory}):
\begin{equation}
    \hhat^{(t+1)} = \qhat\cdot \hhat^{(t)},
    \label{eq:diffusion_step_fourier_space}
\end{equation}
where $\qhat$ is the Fourier coefficients of the diffusion kernel for the irrep $\lambda$.
Therefore, the updates of the Fourier coefficients of $h$ are local: they depend only on the value of the coefficients of $q$ and $h$ at irrep $\lambda$.
The unitary normalization enforced by quantum computation is given replacing
$
    \hat{h}_{\lambda} \mapsto \sqrt{\frac{d_\lambda}{n!}}\hhat.
$
The convolution theorem is rescaled correspondingly:
\begin{equation}
    \widehat{\left(q\star h\right)}_{\lambda} = \sqrt{\frac{n!}{d_\lambda}} \hat{q}_{\lambda}\cdot \hat{h}_{\lambda},
    \label{eq:conv_fourier_space_unitary}
\end{equation}
where $\hat{q}_{\lambda}$ and $\hat{h}_{\lambda}$ are defined here with their unitary normalization prefactor absorbed. 

We consider the following diffusion kernel,
\begin{equation}
    q(\tau)=
    \begin{cases}
        p\;&\tau=e\\
        (1-p)/\binom{n}{2}\; &\tau\in (2,1^{n-2})\\
        0\;& \text{otherwise}
    \end{cases}
    \label{eq:diffusion_kernel}
\end{equation}
where $p$ is the probability of staying applying the identity permutation, $\binom{n}{2}=n(n-1)/2$ is the number of transpositions, and $(2,1^{n-2})$ denotes the conjugacy class of transpositions\footnote{A fundamental result in the representation theory of finite groups is that the number of conjugacy classes equals the number of irreps. For $\sn$, both are indexed by integer partitions of $n$, allowing us to use the same notation for cycle typees and irreps.}. 
Although simple, this diffusion kernel is ergodic: because transpositions form a generating set of $\sn$, a random walk driven by $q$ will, given sufficient steps, reach every permutation~\cite{diaconis1988group} (equivalently, every node of the Cayley graph), and is therefore able to build complex probability distributions.

Furthermore, $q$ is a \textit{class function}, meaning its value is constant over the conjugacy classes of $\sn$: $q(\sigma \tau \sigma^{-1})=q(\tau)$ for all $\tau,\sigma\in\sn$. By Schur's lemma (cf.~App.~\ref{app:representation_theory}), the Fourier transform of any class function is strictly block-diagonal, evaluating to a scalar multiple of the identity matrix for each irrep $\lambda$:
\begin{equation}
    \qhat=c_\lambda\mathbb{I}_\lambda,
    \label{eq:diffusion_kernel_fourier_transform}
\end{equation}
where $\mathbb{I}_\lambda$ is the $d_\lambda\times d_\lambda$ identity matrix and $c_\lambda\in\mathbb{R}$. 

Consequently, the diffusion process we construct does not mix different irreps. 
Because the kernel commutes with the group action, this diffusion step acts as an isotropic, \textit{group-equivariant convolution}. This perfectly mirrors the equivariant layers foundational to Geometric Deep Learning. 
However, rather than operating on tractable, low-dimensional data structures like graphs or images, the convolution here is applied to a factorially large space, where probability distribution and its Fourier spectrum resides.

Applying the Fourier transform and computing the trace on both sides of Eq.~\eqref{eq:diffusion_kernel_fourier_transform} yields
\begin{equation}
    c_\lambda=p+(1-p)\frac{\chi_\lambda((2,1^{n-2}))}{d_\lambda},
    \label{eq:c_lambdas}
\end{equation}
where $\chi_\lambda((2,1^{n-2}))\coloneq \Tr[\irrep(\sigma)]$ for $\sigma\in(2,1^{n-2})$, is the \textit{character} of irrep $\irrep$ evaluated on transpositions.

To implement multiple consecutive diffusion steps $d$, we convolve $q$ with $h$ exactly $d$ times, denoted $h^{(t+d)}(\sigma)=(q^{\star\, d}\,\star\, h^{(t)})(\sigma)$, or equivalently in Fourier space:
\begin{equation}
    \hhat^{(t+d)} = c_\lambda^{d}\cdot \hhat^{(t)}.
\end{equation}
In practice, the number of diffusion steps $d$ between each conditioning operation can be treated as a hyperparameter that is empirically tuned. 

\subsubsection{Block-encoding}
\label{sec:diffusion:block_encoding}
To efficiently implement the diffusion operation on a quantum computer we utilize the QFT and the Fourier space representation of the convolution operation.
Specifically, we embed the non-unitary diffusion kernel for $d\geq 1$ diffusion steps, expressed in Fourier space (regular representation) as
\begin{equation}
    \hat{q}^d \coloneq \bigoplus_{\lambda\vdash n} c^d_\lambda\mathbb{I}_\lambda,
\end{equation}
into a larger unitary operator $D$, which we call the \textit{diffusion operator}, via block-encoding:
\begin{equation}
    D = \begin{pmatrix}
    \hat{q}^d & \star \\
    \star & \star 
    \end{pmatrix}.
    \label{eq:diffusion_operator_fourier_space}
\end{equation}
No rescaling is required to ensure $D$ is unitary, because the singular values of $\qhat$ satisfy $|c_\lambda|\leq 1$ ( App.~\ref{app:diffusion}).

Applying this unitary to the system (or model) register $S$ and an ancillary register $A$ yields:
\begin{equation}
    D \ket{\hat{\psi}^{(t)}}_S\ket{0}_A = \hat{q}^d\ket{\hat{\psi}^{(t)}}_S\ket{0}_A+\ket{\perp},
\end{equation}
where $\ket{\hat{\psi}^{(t)}}$ is the quantum model in Fourier space at time $t$, $\ket{0}_A$ marks the successful subspace, and $\ket{\perp}$ represents the garbage state accumulated in the orthogonal subspace.

The overall efficiency of block-encoding hinges on the efficient construction of an oracle $O_c$ that loads the diffusion coefficients $c_\lambda^d$ conditioned on the irrep label $\lambda$ stored in the Fourier register. Depending on the problem size $n$, we employ two implementation strategies  (detailed in App.~\ref{app:block_encoding_oracle_implementation}):
\begin{itemize}
    \item \textit{Quantum Read-Only Memory (QROM):} For small to moderate $n$ (e.g. 15-30), we can classically pre-compute coefficients and load them via a lookup table:
    \begin{equation}
        \ket{\lambda}\ket{0}_c \xrightarrow{\text{QROM}} \ket{\lambda}\ket{c_\lambda^d}_c.
    \end{equation}
    \item \textit{Coherent Arithmetic:} For large $n$, where the above method becomes cost-prohibitive, we compute $c_\lambda$ on-the-fly, since closed-form formulas exist (App.~\ref{app:block_encoding_oracle_implementation}), using quantum arithmetic circuits:
    \begin{equation}
        \ket{\lambda}\ket{0}_c\ket{0}_{w} \xrightarrow{\text{Arith.}} \ket{\lambda}\ket{c_\lambda^d}_c\ket{w_\lambda}_{w}.
    \end{equation}
\end{itemize}
Here, we omit the $\ket{ij}$ component of the Fourier basis for brevity; $\ket{0}_c$ is a fixed-point precision register for $c_\lambda^d$, and $\ket{w}$ is a temporary workspace that is subsequently uncomputed.

\paragraph*{Amplitude Encoding via Controlled Rotation.}
Once $c_\lambda^d$ is loaded into the ancillary register $\ket{c_\lambda^d}_c$ with either of the two methods, we encode it into the state's amplitude. We introduce a single-qubit ancilla $\ket{0}_A$ and apply a rotation $R_y(\theta_\lambda)$ controlled by the register $\ket{c_\lambda^d}_c$, where the angle is chosen such that $\cos(\theta_\lambda/2) = c_\lambda^d$:
\begin{equation}
    \ket{\lambda}\ket{c_\lambda^d}_c \ket{0}_A \xrightarrow{\text{C-}R_y} \ket{\lambda}\ket{c_\lambda^d}_c \left( c_\lambda^d \ket{0}_A + \sqrt{1-|c_\lambda^d|^2}\ket{1}_A \right).
\end{equation}
Finally, we uncompute the coefficient register $\ket{c_\lambda^d}_c$ by applying the adjoint of the QROM or arithmetic unitary. 
Conditioned on measuring $\ket{0}_A$, the effective action on the system subspace  is the non-unitary multiplication by the diffusion kernel:
\begin{equation}
    (\bra{0}_A \otimes \mathbb{I}) D (\ket{0}_A \otimes \mathbb{I}) = \sum_\lambda c_\lambda^d \ket{\lambda}\bra{\lambda} \equiv \hat{q}^d.
\end{equation}
This completes the block-encoding of the diffusion operator.

\begin{figure}[t]
    \centering
    \includegraphics[width=0.8\linewidth]{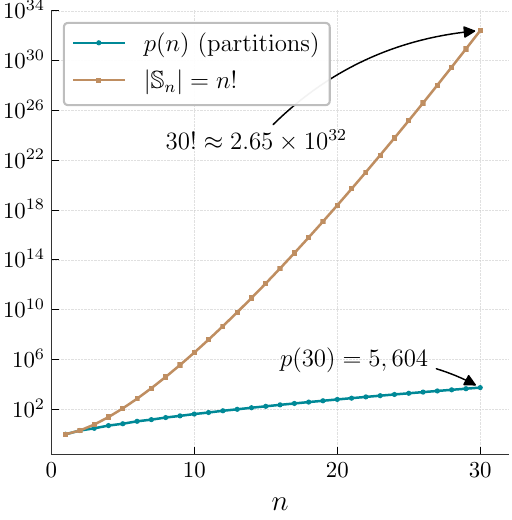}
    \caption{\textit{Scaling of the partitions of $n$ and size of $\sn$}. The number of partitions $p(n)$ is equal to the number of irreps (and conjugacy classes) of $\sn$. While for big $n$, $p(n)$ scales super-polynomially according to the asymptotic formula in Eq.~\eqref{app:eq:ramanujan}, for moderate values of $n$ $p(n)$ is significantly smaller than $n!$. As illustrated, it is possible to classically pre-compute all $c_\lambda$ ($p(n)$ values in total) for moderate $n$ since we have closed formulas, allowing for an efficient block-encoding through QROM.}
    \label{fig:partitions_of_n_vs_n}
\end{figure}

The target state, marked by $\ket{0}_A$, after $d$ consecutive diffusion steps is 
\begin{equation}
    \ket{\hat{\psi}^{(t+d)}} \coloneqq \frac{1}{\norm{\hat{q}^d \ket{\hat{\psi}^{(t)}}}}\sum_{\lambda\vdash n}\sum_{i,j=1}^{d_\lambda} c_\lambda^d\left [\hat{h}^{(t)}_{\lambda}\right]_{ij}\ket{\lambda ij}
\end{equation}
In the following Sec.~\ref{subsec:diffusion_success_prob}, we investigate the scaling of the overlap, or success probability $p_s^{(t)} \coloneq \norm{\hat{q}^d\ket{\hat{\psi}^{(t)}}}^2$ in order to show that our algorithm has a bounded failure probability, and is therefore efficient.

\subsubsection{Scalability: success probability}\label{subsec:diffusion_success_prob}
In addition to the efficiency of the block-encoding implementation which we established above, we investigate the probability of success of each diffusion operation to ensure the scalability of the algorithm. 
Our results are summarized in the following claims.

\begin{claim}[Success probability at $t=0$]\label{claim:ps_t0}
    Let $p_s^{(0)}$ denote the probability of successfully applying the diffusion operator $D$ defined in eq.~\eqref{eq:diffusion_operator_fourier_space} (i.e., measuring $\ket{0}_A$ on the ancillary register).
    For the diffusion kernel defined in eq.~\eqref{eq:diffusion_kernel} the success probability for the initial step is:
    \begin{equation}
        p_s^{(0)} \coloneqq \left\| \hat{q} \ket{\hat{\psi}^{(0)}} \right\|^2 = p^2 + \frac{2(1-p)^2}{n(n-1)}.
        \label{eq:prob_success_init}
    \end{equation}
    (See App.~\ref{app:ps_proof_t0} for the derivation).
\end{claim}
Since for the diffusion process at $t=0$ we can compute $p_s$ exactly, we use standard (Grover) amplitude amplification to efficiently amplify the success to nearly 1 using $\mathcal{O}(1/\sqrt{p_s})$ iterations. 
While we cannot compute exactly the success probability at an arbitrary time $t$, where multiple diffusion and conditioning steps have occurred, we prove the following lower bounds.

\begin{claim}[Success probability lower bounds]\label{claim:ps_arbitrary}
    Let $p_s^{(t)} = \norm{\hat{q}^d \ket{\hat{\psi}^{(t)}} }^2$ be the probability of successfully applying $d$ steps of diffusion (eq.~\eqref{eq:diffusion_operator_fourier_space}) at time $t$ to a quantum model of the form defined in eq.~\eqref{eq:quantum_model_fourier_time_t} that has undergone multiple diffusion and conditioning steps. Then we have the following lower bounds.
    \begin{enumerate}
        \item Lazy Walk Regime ($p > 1/2$): The success probability is lower-bounded by a constant independent of $n$:
        \begin{equation}
            p_s^{(t)} \ge (2p - 1)^{2d}.
        \end{equation}
        \item General Rational Walk: For any rational $p = a/b \in (0,1)$ (where $c_\lambda \neq 0$), the success probability is lower-bounded by an inverse polynomial in $n$:
        \begin{equation}
            p_s^{(t)} \ge \left(\frac{4}{b^2 n^4}\right)^d.
        \end{equation}
    \end{enumerate}
    (See App.~\ref{app:ps_proof_arbitrary} for the derivation).
\end{claim}

\subsection{Conditioning Operation}
\label{sec:conditioning}
We consider conditioning operations that update the quantum model upon data observation by implementing Bayes' rule. In particular, the update of the probability distribution $h^{(t)}(\sigma)$ is performed in \textit{direct space}. Indeed, conditioning in Fourier space requires computing Kronecker (Clebsch-Gordan) coefficients for the symmetric group (App.~\ref{app:sec:conditioning_in_Fourier}), a problem known to be \#P-hard~\cite{fischer2020computationalcomplexityplethysmcoefficients,Ikenmeyer_2017}. Instead in direct space, up to normalization, the Bayes update step is as simple as a point-wise product: 
\begin{equation}
\label{eq:bayes}
    h^{(t+1)}(\sigma) = h^{(t)}(\sigma | \varphi) \propto h^{(t)}(\sigma)\, h(\varphi | \sigma),
\end{equation}
where $\varphi:\sn\to\{0,1\}$ is a data-encoding function flagging with $\varphi(\sigma)=1$ permutations $\sigma$ \emph{consistent} with observations, and $h(\varphi | \sigma)$ is a likelihood function, where we account for uncertainty and/or trust in the data acquisition process. This is modeled by introducing a parameter $s \in (0.5,1]$, quantifying the probability of correctly flagging permutations as ``consistent''. More precisely, we consider
\begin{equation}
\label{eq:def_likelihood}
h(\varphi| \sigma) = \begin{cases}
        s &\text{if}\;\; \varphi(\sigma) = 1\\
        1-s &\text{otherwise}.\\
    \end{cases}
\end{equation}
In this setting, uncertainty is reflected in \emph{keeping} inconsistent permutations in the model, and only reducing their probability by a factor proportional to $1-s$ in the posterior. 
Full trust on the data can be obtained by setting $s=1$, in which case permutations $\sigma_{nc}$ not fully consistent with empirical observations will be removed from the support of the model, i.e. $h^{(t+1)}(\sigma_{nc}) = 0$. Such a condition is often referred to as a \emph{hard} likelihood.

\subsubsection{Block-encoding}

Given a data-encoding function $\varphi$, we implement Bayes' rule on a quantum computer by embedding the operation into a larger, data-dependent unitary operator $C(\varphi)$, which we call the \emph{conditioning operator}, via block-encoding:
\begin{equation}
    C(\varphi) = \begin{pmatrix}
     \,c_\varphi & \star \\
    \star & \star 
    \end{pmatrix},
    \label{eq:bayes_operator_block_encoding}
\end{equation}
where $c_\varphi$ is a linear operator, diagonal in the computational basis, and is defined as 
\begin{equation}
    c_\varphi = \sum_{\sigma \in\sn}h(\varphi|\sigma)\, \ket{\sigma}\bra{\sigma}\,.
\label{eq:bayes_operator_non_unitary}
\end{equation}
The unitary operator $C(\varphi)$ acts now on both the model/system register $S$ and the ancillary register $A$:
\begin{equation}
   C(\varphi) \ket{\psi^{(t)}}_S\ket{0}_A = c_\varphi\ket{\psi^{(t)}}_S\ket{0}_A+\ket{\perp},
\end{equation}
with notation analogous to the previous section.

Note that, the operator $C(\varphi)$ is univocally determined by the function $\varphi$, and does not depend on the current quantum model $\ket{\psi^{(t)}}$. Indeed, as long as the function $\varphi$ can be efficiently computed, the coherent arithmetic strategy outlined in Section \ref{sec:diffusion:block_encoding} can be analogously applied here to block-encode $c_\varphi$. The workflow unfolds as follows: first, we compute coherently $\varphi(\sigma)$ using quantum arithmetic
\begin{equation}
    \ket{\sigma}\ket{0}_c\ket{0}_{w} \xrightarrow{\text{Arith.}} \ket{\sigma}\ket{\varphi(\sigma)}_c\ket{w_\sigma}_{w},
\end{equation}
where $\ket{w}$ is a temporary workspace that is subsequently uncomputed. Then we introduce a single qubit ancilla $\ket{0}_A$, and apply the rotation $R_y(\theta_{s})$ controlled by the state $\ket{\varphi(\sigma)}_c = \ket{1}$, such that $\cos(\theta_s/2) = s$ and subsequently $R_y(\theta_{1-s})$ controlled by the state $\ket{\varphi(\sigma)}_c = \ket{0}$, with analogous definition of $\theta_{1-s}$. This yields the overall transformation
\begin{equation}
\begin{aligned}
    \ket{\sigma}&\ket{\varphi(\sigma)}_c\ket{0}_A \xrightarrow{\text{C-}R_y} \\&\ket{\sigma}\ket{\varphi(\sigma)}_c \left( h(\varphi|\sigma)\, \ket{0}_A + \sqrt{1-h(\varphi|\sigma)^2}\,\ket{1}_A \right).
\end{aligned}
\end{equation}
Finally, we uncompute the registers $\ket{\varphi(\sigma)}_c\ket{w_\sigma}_{w}$ to disentangle them from the system. The effective action on the system subspace, conditioned on measuring $\ket{0}_A$, is the non-unitary application of Bayes rule:
\begin{equation}
    (\bra{0}_A \otimes \mathbb{I}) C(\varphi) (\ket{0}_A \otimes \mathbb{I}) = \sum_{\sigma\in \sn} h(\varphi|\sigma) \ket{\sigma}\bra{\sigma} \equiv c_\varphi.
\end{equation}

While implementing $C(\varphi)$ in this way is already efficient, as we move away from the broad generality of the previous discussion, and tailor our analysis to \emph{common types} of empirical observations, it is possible to significantly reduce the algorithmic cost of implementing $C(\varphi)$. We outline this in the following section. 

\subsubsection{The reorder-update approach}

\begin{figure}[t]
    \centering
    \includegraphics[width=\linewidth]{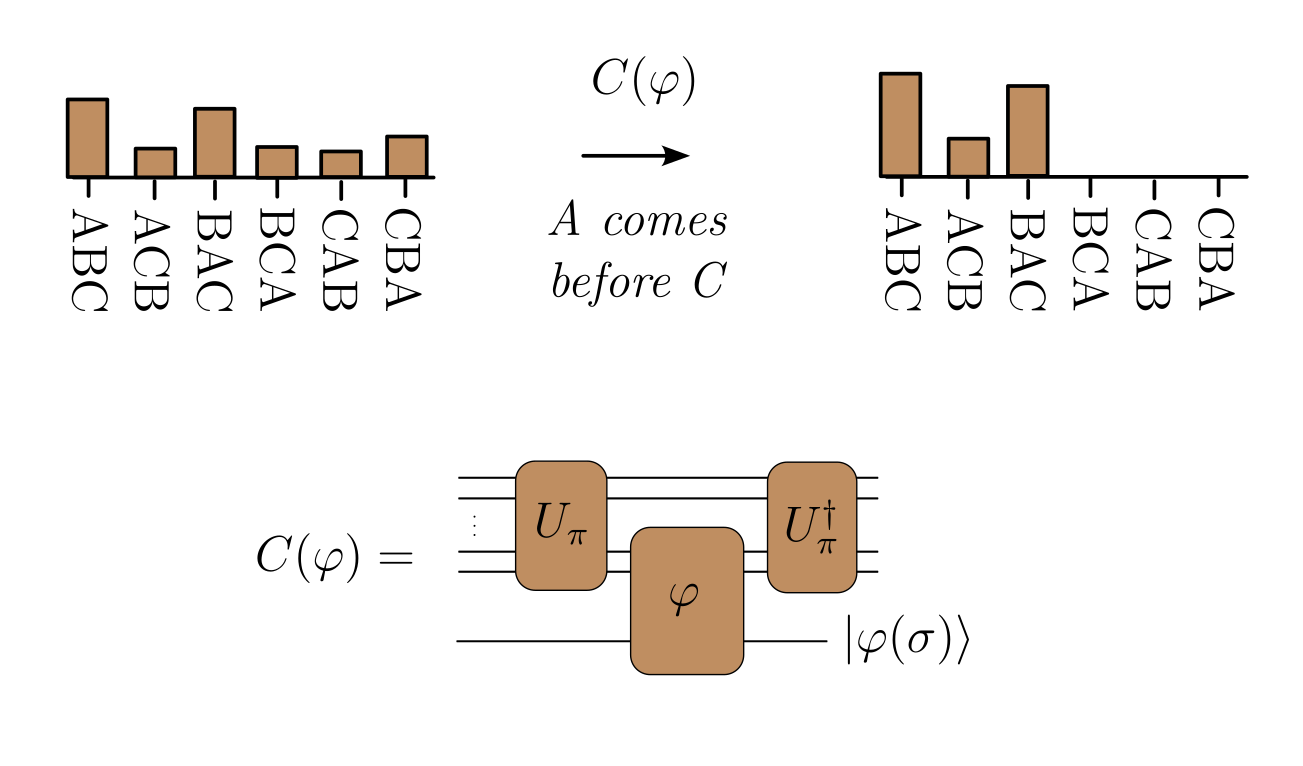}
    \caption{\textbf{Reorder-update approach:} The figure summarizes the reorder-update approach focusing on the example of partial rankings with a hard likelihood. Using the same canonical ordering of Fig.~\ref{fig:markov}, the observation \emph{A comes before C} is an order $k=1$ ranking, comparing element $i_1=1$ and $i_2=3$. Conditioning is achieved reordering the elements by mapping $i_1\to 2$, i.e. choosing $\pi$ s.t. $c(\pi)_1 = 2$, computing $\varphi$ coherently, and finally restoring the canonical ordering. If a hard likelihood is used, the qubit storing $\ket{\varphi(\sigma)}$ can be directly used in the post-selection phase to implement $c_\varphi$, and no additional ancillary registers are required.}
    \label{fig:partial_assigments_conditioning}
\end{figure}

Here, we consider two major types of empirical observations, each fixing different constraints on the structure of the permutations $\sigma$ of $n$ objects.  

On the first type, we observe the position of a subset of $k \ll n$ objects. 
In the language of the example in Fig. \ref{fig:markov}, the observation \emph{element $A$ is first}, is of this kind with $k=1$, as only one element is observed to be in a determined position, regardless of the others. We call this class of observations \emph{partial assignments}, as we partially assign the position of objects.
More generally, if we let $c(\sigma)_i$ refer to the new position of the object ranked $i$ in the canonical ordering after the permutation $\sigma$, and let $j$ be the empirically observed position, such constraint can be expressed as
\begin{equation}
\label{eq:partial_assignment}
    c(\sigma)_{i_1}=j_1 \;\wedge \; c(\sigma)_{i_2}=j_2 \;\wedge \; \dots\; \wedge c(\sigma)_{i_k}=j_k\,,
\end{equation}
so that the corresponding data-encoding function $\varphi$ will evaluate to 1 when Eq.~\eqref{eq:partial_assignment} is satisfied and 0 otherwise. 

On the second class instead, we observe the relative order of a subset of $k+1 \ll n$ objects.
In the language of Fig. \ref{fig:markov}, the observation \emph{element $A$ comes before element $C$}, is a partial ranking with $k=1$, as only the relative order of two elements, $A$ and $C$, is observed, regardless of their position or other elements.
More generally, this can be expressed as 
\begin{equation}
\label{eq:partial_ranking}
    c(\sigma)_{i_1} > c(\sigma)_{i_2} \; \dots\; > c(\sigma)_{i_k} > c(\sigma)_{i_{k+1}}\,,
\end{equation}
with analogous definition of $\varphi$.

While the constraints posed by Eq.~\eqref{eq:partial_assignment} and Eq.~\eqref{eq:partial_ranking} are intuitively understood in the \emph{Cauchy encoding} of permutations $c(\sigma)$, which explicitely stores the permuted positions of objects, our model requires the \emph{Lehmer} encoding $\ell (\sigma)$, which has different properties (see App. \ref{app:sec:conditioning:lehmer}). 

Indeed, naively enforcing a partial assignment for general elemets $i_1,\dots i_k$ in our setting involves bookkeeping operations on all entries $\ell (\sigma)_i$ and thus requires complex controlled operations to be implemented. 
Nonetheless, specific sets of indices $i_1,\dots i_k$ exist for which both constraints listed above can be still locally enforced. As detailed in Appendix \ref{app:sec:conditioning}, for partial assignments, this is holds when all indices are at the beginning of the string, namely $i_l \in \{1,\dots,k\} \forall l$, while for partial ranking when all indices are at the end of the string, namely $i_l \in \{n-k,\dots,n\} \forall l$. 

This leads to our general \emph{reorder-update} approach. First, we implement a unitary transformation $U_\pi$, where the permutation $\pi$ maps all relevant indices to the beginning (or end respectively) of the Lehmer code string. This is achieved by implementing $\ket{\ell(\sigma)} \to U_\pi\ket{\ell(\sigma)} = \ket{\ell(\sigma \pi)}$.
As a second step, we implement $C(\varphi)$ as detailed in the previous section, with the with the difference that now $\varphi$ can be computed involving only a small subset of qubits and without a temporary workspace $\ket{w}$. If we use a \emph{hard} likelihood, $C(\varphi)$ becomes a projector, which is realized upon measurement of $\ket{1}$ on the register containing $\varphi
(\sigma)$, and without the need of further auxiliary qubits or uncomputation steps. 
Finally, $U^\dagger_\pi = U_{\pi^{-1}}$ is applied to the system to restore the canonical ordering of objects.

The overall cost of this procedure is summarized below.

\begin{claim}[Cost of the reorder-update approach] Given an order $k$ constraint of type Eq.~\eqref{eq:partial_assignment} or Eq.~\eqref{eq:partial_ranking}, the reordering unitary $U_\pi$ can be implemented in $\mathcal{O}(kn^2\log n)$ depth, while the computation of $\varphi$ requires at most $\mathcal{O}(k^2\log^2 n)$ depth.
\end{claim}

All details regarding the specific choice of $\pi$ and the computation of $\varphi$ are given in Appendix \ref{app:sec:conditioning}.

\subsubsection{Scalability: success probability}

In addition to the efficiency of the implementation of $C(\varphi)$ established in the previous section, here we study the success probability  $p_s$ of a general Bayes update step. The results are summarized in the following claim.

\begin{claim}[Success probability]
\label{claim:succ_prob_conditioning}
Let $p_s$ denote the probability of successfully applying the conditioning operator $C(\varphi)$ defined by Eq.~\eqref{eq:bayes_operator_block_encoding}. For a likelihood defined as in Eq.~\eqref{eq:def_likelihood} the success probability is:
    \begin{equation}
\label{eq:succ_prob_conditioning}
        p_s \coloneqq \left\| c_\varphi \ket{\psi^{(t)}} \right\|^2 = h(\varphi)^2\frac{N^{(t+1)}}{N^{(t)}},
    \end{equation}
where $h(\varphi)$ denotes the total proability of generating consistent permutations in the quantum model, and $N^{(t)}$ denote the model normalization at step $t$.  
\end{claim}

Since $p_s$ depends on the specific probability distribution $h^{(t)}(\sigma)$ encoded in the quantum model as well as the observed data encoded by $\varphi$, giving a general lower bound in terms of $n$ is unrealistic. However, thanks to Eq.~\eqref{eq:succ_prob_conditioning}, we note a general trend: when the diffusion step is well informed and able to capture the data distribution, $h(\varphi)$ is large, and the conditioning step succeds with high probability. Contrarily, when the diffusion step is poorly informed and unable to capture the data real distribution, the opposite happens, and the conditioning step will fail with high probability.

Hence, while Eq.~\eqref{eq:succ_prob_conditioning} does not give rigourous scalability guarantees, it suggests that useful and accurate models will perform well, and that a vanishing success probability in the conditioning step will be mainly suffered by less useful, worse-performing models.

\subsection{Amplitude Amplification}\label{sec:amplitude_amplification}
Both diffusion and conditioning are non-unitary operations on the model register $S$. They are implemented via unitary operators $D$ (Sec.~\ref{sec:diffusion}) and $C$ (Sec.~\ref{sec:conditioning}), which act on an expanded Hilbert space comprising $S$ and an ancillary $A$ register.
To maintain coherence during the model evolution, we postpone any ancillary measurements or amplitude amplification until the final time $T$, after all interleaved diffusion and conditioning steps have been applied.

Each application of $D$ or $C$ needs an additional ancilla, accumulating a $T$-qubit register where the target subspace is marked by the all-zero bitstring $\ket{0^T}_A$.
Let $\mathcal{P}$ be the lower bound on the total success probability at time $T$ $p_\textbf{tot}\coloneqq \norm{\Braket{0^T_A|\psi^{(T)}}}^2$: 
\begin{equation}
\mathcal{P} = \prod_{t=1}^T p_{\text{diff}}^{(t)} \cdot p_{\text{cond}}^{(t)} = \left((2p-1)^{2d}\cdot p_{\text{cond}} \right)^T,
\label{eq:total_prob_succ_lb}
\end{equation}
where $p_{\text{diff}}^{(t)}$ and $p_{\text{cond}}^{(t)}$ are the lower bounds of the individual diffusion and conditioning steps, from Claim~\ref{claim:ps_arbitrary} and Claim~\ref{claim:succ_prob_conditioning}, respectively.
To boost this probability to $1$, we can employ standard amplitude amplification~\cite{Brassard:2000xvp}. In the worst case, this requires $\mathcal{O}(1/\sqrt{\mathcal{P}})$ measurement trials and total algorithm reruns (Theorem 4 of~\cite{Brassard:2000xvp}). 
Alternatively, Fixed-Point Amplitude Amplification (FPAA)~\cite{Yoder:2014sls} enables coherent, deterministic state preparation without intermediate measurements. Achieving a target error tolerance $\delta$ (i.e., a success probability of $1-\delta$) requires a circuit depth scaling as
$\mathcal{O}\left( \log(1/\delta)/\sqrt{\mathcal{P}} \right)$.
The choice between these methods depends on the hardware budget and the requirements of the downstream ML task (Sec.~\ref{sec:model_use}): FPAA has larger constant prefactors requiring deeper circuits but ensures coherent state preparation, while standard AA requires repeated executions but permits shallower circuits.

For our algorithm to remain \textit{asymptotically} scalable with $n$, the total number of steps should scale as $T\in\mathcal{O}(\log n)$. 
It is important to emphasize that $T$ refers to the \textit{total} number of \textit{interleaved} diffusion and conditioning pairs, not the overall circuit depth, which is always dominated by the $\sn$-QFT. 
While the lower bound in Eq.~\eqref{eq:total_prob_succ_lb} may not be optimal, deriving a tighter bound remains an open question.

\section{Using the quantum model}\label{sec:model_use}
The last section described how we can prepare the probabilistic model $h^{(T)}(\sigma)$ at time $T$, suggested in~\cite{kondor2011non, JMLR:v10:huang09a}, via amplitude encoding in a quantum state due to the remarkable properties of the $\sn$ QFT. 
This encoding enables us to lift the bandlimiting restrictions required on classical computers, and allows us to efficiently move between Fourier and direct space, rather than performing all computations in Fourier space. But the encoding also dictates the scope of how the model can be used, which is very different from the classical case, mirroring a \textit{fundamental dichotomy} between classical and quantum probabilistic modeling. Classically, the (bandlimited) model can only be used for inference tasks performed by computing lower-order Fourier coefficients---evaluating low-order marginals of $h^{(t)}(\sigma)$---while sampling from $h^{(t)}(\sigma)$ is hard. When Fourier coefficients are encoded in amplitudes of a quantum state, they cannot be scalably estimated this way. Instead, the quantum state gives us a certain kind of \textit{sampling access} in direct and Fourier space via measurements, which is very attractive for \textit{generative modeling}. In this section, we characterise this sampling access in detail, and present a few promising strategies for utilizing the prepared quantum states in downstream machine learning tasks. 

More precisely, we will discuss how using the \textit{amplitude encoding} of $h^{(t)}(\sigma)$ discussed so far, where the amplitudes $\psi_\sigma^{(T)}$ of the prepared state $\ket{\psi^{(T)}}$ encode directly $\ell^2$-normalized posterior probabilities $h^{(t)}(\sigma)$, and measurements sample from the slightly attenuated distribution $p(\sigma) \propto |h^{(t)}(\sigma)|^2$. 
While this may be useful for some tasks, we will introduce an alternative encoding option, which we will call \textit{Born encoding}, where the amplitudes represent the square root of the probability distribution, $\psi_\sigma^{(T)}=\sqrt{h^{(t)}(\sigma)}$, so that measurements are sampled from the posterior $p(\sigma) \propto h^{(t)}(\sigma)$. As we discuss below, this turns the diffusion process into a non-Markovian evolution, which we can still implement using the same algorithm presented in Sec.~\ref{sec:model}. Both encodings allow us to also sample in Fourier space, which might prove useful for some applications.

Finally, we give some ideas on how one might be able to attenuate the probabilities by applying a higher-order polynomial onto the amplitudes. This would enable optimisation tasks such as \textit{maximum a posteriori} (MAP) inference; identifying most likely configuration (permutation) according to the model. MAP inference is a notoriously challenging optimization problem, particularly for multimodal or discrete model distributions. 
Even more, we get the ability of MAP estimates of certain marginal distributions, considered one of the hardest inference tasks, for free by measuring only part of the qubits.

\subsection{Generative Modeling: Amplitude Encoding}
The quantum algorithm described in Sec.~\ref{sec:model} prepares the state (omitting the ancillary register):
\begin{equation}
    \ket{\psi^{(T)}}=\sum_{\sigma\in\sn} \psi_\sigma^{(T)}\ket{\sigma},
\end{equation}
where the amplitudes are positive and encode 
the posterior probability distribution over the permuted objects at time $T$,
\begin{equation}
    \psi_\sigma^{(T)}=\frac{h^{(T)}(\sigma)}{\sqrt{N^{(T)}}},
\end{equation}
where $N^{(T)}=\sum_\sigma|h^{(T)}(\sigma)|^2$ is the squared $\ell_2$ norm of $h$, required for the normalization of the state.  

As described in Sec.~\ref{sec:model}, the evolution of the model is driven by the interplay between conditioning (Bayes updates based on data) and diffusion (entropy spreading between observations). 
The expressivity of the model is governed by the balance between these two operations: diffusion expands the support of $h$, while conditioning concentrates it in alignment with the observed data.
The number of diffusion steps $d$ applied between conditioning events is a hyperparameter, which needs to be tuned to the task at hand to ensure that the probability distribution $h^{(t)}(\sigma)$ approximates the observed data well.

The measurement of the prepared state on the computational basis yields samples that are distributed according to 
\begin{equation}\sigma\sim\left|\Braket{\sigma|\psi^{(T)}}\right|^2\propto |h^{(T)}(\sigma)|^2.
\label{eq:square_sampling}
\end{equation}
Therefore, the algorithm can be viewed as a \textit{quantum generative model} that can generate new data samples according to a distribution proportional to the squared posterior $|h^{(T)}(\sigma)|^2$. 
Practically, this \textit{sharpening} of the posterior distribution results in enhancing the chance of observing permutations $\sigma$ that have high probability under the posterior, while attenuating those with low---i.e., attenuating the tails of $h^{T}(\sigma)$.

This form of biased sampling from $h^{T}(\sigma)$ can be useful for tasks where the objective is to heuristically find the most probable configurations: most likely identity-to-track configuration or preference ranking of $n$ objects. We discuss this further in Sec.~\ref{sec:qsvt_sharpening}.

We highlight that this property of sharpened sampling (Eq.~\eqref{eq:square_sampling}) is a direct consequence---a \textit{feature}---of our quantum algorithm and is governed by its initialization, which subsequently induces an amplitude encoding of $h^{(t)}(\sigma)$.
Specifically, if at $t=0$ we start with a state encoding the canonical configuration, i.e. setting the identity element of $\sn$, of the objects $\ket{\psi^{(0)}} = \ket{e}$. 
In this setting, which can be viewed as ``online learning'', we evolve the belief state from an initial configuration and inject information from new data (observations) as they come.
The amplitudes of the state at each $t$ are proportional to $h^{(t)}(\sigma)$, emulating the classical Markov Chain dynamics from diffusion and conditioning on the quantum computer.
An alternative encoding of the belief state and the data is presented in the following.

\subsection{Generative Modeling: Born Encoding}\label{sec:born_encoding}
 An alternative setting is where at $t=0$ we already have a dataset $\mathcal{D}=\{\sigma_i\}_{i=1}^N$ of $N$ observed object configurations, such as a preferred ranking of $n$ objects by $N$ users. 
 We can load $\mathcal{D}$ on the quantum computer by preparing what we call an \textit{empirical state},
\begin{align}
   \ket{\psi^{(0)}} \equiv\ket{\mathcal{D}}&= \sum_{i=1}^N\psi_{\sigma_i}^{(0)}\ket{\sigma_i}\\
   &=\sum_{i=1}^N\sqrt{h^{(0)}(\sigma_i)}\ket{\sigma_i},
   \label{eq:square_root_encoding_init}
\end{align}
where the amplitudes encode the empirical probabilities $h^{(0)}(\sigma_i)=\frac{k_i}{N}$, $k_i$ being the number of occurrences of permutation $\sigma_i$. Note that this is a sparse state in the sense that $N\ll n! $. Crucially, using the empirical state chooses \textit{Born encoding} over \textit{amplitude encoding} of the probabilistic model.

With this alternative encoding, and starting with a different initial state, the developed diffusion and conditioning operations remain well defined. 
Conditioning involves the unitary manipulation of Lehmer's code, used to map permutations to computational basis states, and the projection onto the subspace consistent with the new observation (data) as detailed in Sec.~\ref{sec:conditioning}.
Therefore, the conditioning operation implements a consistent Bayes update on the amplitudes of the quantum model $\psi_\sigma^{(t)}$ both in the case where they directly encode the belief state $\psi_\sigma^{(t)}=h(\sigma)^{(t)}/\sqrt{N^{(t)}}$, and in the case of the Born encoding of Eq.~\eqref{eq:square_root_encoding_init}, up to rescaling the likelihood according to the square root: $\sqrt{h(\varphi|\sigma)}$ in Eq.~\eqref{eq:bayes}.

The case of diffusion is more subtle. While it still remains a well-defined operation as constructed in Sec.~\ref{sec:diffusion}, it induces a process that is not Markovian in the sense of linear dynamics governed by a transition matrix (see App.~\ref{app:sec:markov_matrix_form}).
Specifically, the operation becomes a diffusion of \textit{quantum amplitudes} $\psi^{(t)}_\sigma = \sqrt{h^{(t)}(\sigma)}$, where a step is defined by the update:
\begin{equation}
    \hat{\psi}_\lambda^{(t+1)} \coloneqq\frac{\qhat \cdot \hat{\psi}_\lambda^{(t)}}{\mathcal{N}^{(t)}} 
\end{equation}
where $\hat{\psi}_\lambda^{(t)} \coloneq \sqrt{\frac{d_\lambda}{n!}}\sum_{\sigma \in \sn} \psi_\sigma^{(t)}\irrep(\sigma)\ket{\lambda ij}$, and the renormalization, 
\begin{align}
    \mathcal{N}^{(t)}&=\sqrt{\sum_{\mu\vdash n}||\hat{q}_\mu\cdot \hat{\psi}_\mu^{(t)}||^2_F}\\
    &=\sqrt{\sum_{\mu\vdash n} \Tr\left[(\hat{q}_\mu\cdot \hat{\psi}_\mu^{(t)})^\dagger \cdot(\hat{q}_\mu\cdot \hat{\psi}_\mu^{(t)})\right]}
    \label{eq:square_root_encoding_diffusion_norm}
\end{align}
is required to preserve the unit $\ell_2$-norm of the amplitudes (conservation of total probability) since the kernel $q$ is stochastic but not unitary.
This process can equivalently be formulated in direct space. The update consists of a linear convolution followed by the same non-linear renormalization:
\begin{equation}
    \psi^{(t+1)}(\sigma) = \frac{1}{\mathcal{N}^{(t)}} \sum_{\tau \in G} q(\sigma \tau^{-1}) \psi^{(t)}(\tau).
\end{equation}
By Parseval's identity (Plancherel Theorem App.~\ref{app:sec:plancherel_qft}), the normalization factors in both domains are identical.

\paragraph*{Intuitive Meaning.}
The physical intuition remains analogous to the Markov case: the diffusion kernel $q$ acts as a low-pass filter, suppressing high-frequency irreps. In direct space, this corresponds to spreading the support of the probability distribution $h^{(t)}(\sigma) = |\psi^{(t)}_\sigma|^2$. However, because the low-pass filtering is applied to the spectrum of $\sqrt{h}$ rather than $h$, the exact dynamics differ. Since $\sqrt{h}$ represents a spatially smoother function than $h$, its spectrum is more heavily concentrated in the low-frequency irreps. Consequently, while both processes share the uniform distribution as a stationary state (given $d\in\mathcal{O}(n\log n)$ random walk steps~\cite{diaconis:1981_random_walk}), their trajectories differ, implying that the tunable number of diffusion steps to model a given dataset may deviate from the classical Markovian estimate. 
The Fourier spectrum of the amplitudes and that of the probabilistic model are related by generalized non-abelian \textit{auto-convolution}. For more details see App.~\ref{app:sec:diffusion_amplitudes}.

\subsection{Fourier Sampling}
Beyond generative modeling on the computational basis it is possible to also do (weak and strong) \textit{Fourier sampling} \cite{Childs:2008shx} on $\ket{\hat{\psi}^{(T)}}$.
Depending on the encoding of the probability distribution, this procedure yields information about different spectral properties.

\paragraph*{Amplitude Encoding.}
In the case where the state encodes the square root of the probability distribution, $|\psi^{(t)}\rangle = \sum_\sigma \sqrt{h^{(t)}(\sigma)} |\sigma\rangle$, Fourier sampling returns an irreducible representation label $\lambda$ with probability $P_{\text{amp}}(\lambda)$ proportional to the energy of the amplitude spectrum:

\begin{equation}
    P_{\text{amp}}(\lambda) = \left\| \hat{\psi}_\lambda^{(t)} \right\|_F^2,
\end{equation}
where $\| \cdot \|_F$ denotes the Frobenius norm (as in Eq.~\eqref{eq:square_root_encoding_diffusion_norm}).  Since the square-root operation compresses the dynamic range of the distribution, $\sqrt{h}$ is spatially smoother than $h$, resulting in a spectrum $P_{\text{amp}}$ that is more rapidly decaying and concentrated on low-frequency irreps.

\paragraph*{Born Encoding.}
Alternatively, if one prepares a state where amplitudes are directly proportional to the probabilities, $|\phi^{(t)}\rangle = \frac{1}{\|h^{(t)}\|_2} \sum_\sigma h^{(t)}(\sigma) |\sigma\rangle$, Fourier sampling yields the power spectrum of the density itself. The probability of observing $\lambda$ is given by:

\begin{equation}
    P_{\text{Born}}(\lambda) = \frac{1}{\sum_\sigma |h^{(t)}(\sigma)|^2} \left\| \hat{h}_\lambda^{(t)} \right\|_F^2.
\end{equation}
Here, $\hat{h}_\lambda^{(t)}$ is the Fourier transform of the probability distribution $h$. Because $h$ is generally ``spikier'' (sparser) than $\sqrt{h}$, this distribution is broader in Fourier space. 

\paragraph*{Relationship.}
The two sampling outcomes are rigorously linked by the generalized convolution theorem. Since $h(\sigma) = \psi_\sigma \cdot \psi_\sigma$, the spectrum sampled in the probability encoding ($P_{\text{prob}}$) corresponds to the auto-convolution of the spectrum sampled in the amplitude encoding ($P_{\text{amp}}$), governed by the Clebsch-Gordan series of $\sn$ (cf. App.~\ref{app:sec:diffusion_amplitudes}).

\subsection{Maximum-A-Posteriori Estimates}\label{sec:qsvt_sharpening}

In tracking and ranking tasks, a frequent objective is to identify the most probable configurations. Namely, after $T$ training steps of the model, the goal is to extract the set $\Omega_k =\{\sigma_1,\sigma_2,\dots,\sigma_k\}$ containing the $k$ most likely candidates---often referred to as the \textit{top-$k$} configuration, or as \textit{maximum-a-posteriori} estimates. These may represent, for instance, specific object-to-track assignments or partial rankings.

Using classical algorithms, attaining such MAP estimates involves computationally intractable optimization problems, particularly when the model captures correlations between permutations or more complex dependencies~\cite{pmlr-v2-kondor07a,JMLR:v10:huang09a}. 
In simple cases, where the model is factorizable, i.e., it involves statistically independent probabilities for each assignment---equivalently expressed in Fourier space as $\hat{h}_\irrep=0$ for $\lambda < (n-1,1)$---the optimization is a Linear Assignment Problem solvable via the Hungarian algorithm. In practice, these cases are exceptions, and classical heuristics can become computationally infeasible without approximations such as bandlimiting~\cite{pmlr-v2-kondor07a,JMLR:v10:huang09a}.

Instead of mapping this task to an optimization problem, the quantum modeling approach suggests turning it into a sampling task. This would require amplifying the probability of observing permutations $\sigma$ where $h^{(T)}(\sigma)$ is large. 
As this may be of independent interest in quantum generative modeling, we sketch this strategy using more general notation, for the remainder of this section.

The idea is to use the state preparation routine $U$ of the model to block-encode the diagonal matrix of state amplitudes $\psi_\sigma$, and subsequently apply a polynomial filter via Quantum Singular Value Transformation (QSVT) to prepare the sharpened posterior~\cite{Gilyen:2018khw}. 

Suppose we have a quantum state representing a generative model, prepared by a unitary U acting on a system register $S$ and an ancillary register $A_\text{model}$\footnote{So far, we have used $\ket{0}_A$ for the ancillary register needed for the preparation of the quantum model. In this section, we denote it $A_\text{model}$ to differentiate it from the additional ancillary register $A$ that QSVT requires.}
\begin{equation}
\ket{\Psi} \coloneqq U\ket{0}_S\ket{0}_{A_{\text{model}}} = \left(\sum_\sigma \psi_\sigma \ket{\sigma}_S\right)\ket{0}_{A_{\text{model}}},
\label{eq:unitary_state_prep_assumption}
\end{equation}
where we assume efficient uncomputation of the ancillary register back to $\ket{0}_{A_\text{model}}$, as is the case when using amplitude amplification (detailed in Sec.~\ref{sec:amplitude_amplification}).

To perform QSVT, we first construct a unitary $W$ that block-encodes the diagonal matrix of these posterior amplitudes (omitting the time step label for brevity):
\begin{equation}
( \mathbb{I}_S\otimes\bra{0}_A)W(\mathbb{I}_S\otimes\ket{0}_A) = \sum_\sigma \psi_\sigma\ketbra{\sigma}{\sigma}_S.
\label{eq:block_encoding_posterior_amplitudes}
\end{equation}
Note that $W$ acts on the system register $S$ and a new ancillary register $A$. Because this register must hold both the copied state and the workspace for the inverse state preparation, its size is $n_A = n_S + n_{A{\text{model}}}$.
In practice, the uncomputed ancillary register $A_\text{model}$ used previously to construct $U$ can be reused and extended to serve as $A$.

\begin{claim}[Block-Encoding of Posterior Amplitudes]
    Let $U_A$ be an instantiation of the state preparation unitary, acting on the ancillary register $A$ such that $U_A\ket{0}_A=\sum_\sigma \psi_\sigma \ket{\sigma}_{A_\textbf{copy}}\ket{0}_{A_\text{work}}$, where we denote the joint state simply as $\sum_\sigma\psi_\sigma \ket{\sigma}_A$ for brevity. 
    Furthermore, let $V_\text{copy}$ denote the transversal CNOT operation, acting bit-wise with $S$ as control and $A$ as target, that performs the mapping $\ket{\sigma}_S\ket{0}_A\mapsto \ket{\sigma}_S\ket{\sigma}_A$. 
    The unitary operator defined by 
    \begin{equation}
        W\coloneqq (\mathbb{I}_S\otimes U_A^\dagger)V_\text{copy},
    \end{equation}
    is a $(1,n_A,0)$-block-encoding, following the notation of~\cite{Gilyen:2018khw}, of the diagonal operator $\tilde{A}=\sum_\sigma \psi_\sigma \ketbra{\sigma}{\sigma}_S$ corresponding to the posterior amplitudes defined in Equation~\eqref{eq:block_encoding_posterior_amplitudes}.
\end{claim}
\begin{proof}
Consider the action of $W$ on the joint basis state $\ket{\sigma}_S\ket{0}_A$.
First, we apply the CNOT cascade operation 
\begin{equation}
V_\text{copy}\ket{\sigma}_S\ket{0}_A=\ket{\sigma}_S\ket{\sigma}_A.     
\end{equation}
Next, applying the inverse state preparation $(\mathbb{I}_
S\otimes U_A^\dagger)$ yields 
\begin{equation}
(\mathbb{I}_S \otimes U^\dagger_A) \ket{\sigma}_S \ket{\sigma}_A = \ket{\sigma}S \otimes (U^\dagger_A \ket{\sigma}_A).
\end{equation}
To find the effective operator on the system register, we project the ancilla register back to the zero state $\ket{0}_A$,
\begin{align}
\bra{0}_A W \ket{\sigma}_S \ket{0}_A &= \ket{\sigma}_S \bra{0}_A U^\dagger_A \ket{\sigma}_A\nonumber\\
&= \ket{\sigma}_S \left( \bra{\sigma}_A U_A \ket{0}_A \right) \nonumber\\
&= \psi_\sigma^* \ket{\sigma}_S \nonumber\\
&=\psi_\sigma \ket{\sigma}_S,
\label{eq:block_encoding_posterior_amplitudes_proof}
\end{align} 
where the final equality holds because the amplitudes $\psi_\sigma$ are real for quantum generative models, and $\Braket{0|U^\dagger|\sigma}_A=\psi_\sigma$. 
Since the equality in Equation~\eqref{eq:block_encoding_posterior_amplitudes_proof} holds for all basis states $\ket{\sigma}_S\ket{0}_A$, we conclude that the projected operator is the diagonal matrix $\tilde{A} = \sum_\sigma \psi_\sigma \ket{\sigma}\bra{\sigma}_S$, completing the proof.
\end{proof}

After applying $U$, we have in the system register $S$ the quatum model $\ket{\Psi}$ and prepare the joint state:
\begin{equation}
    \ket{\Phi} \coloneqq \ket{\Psi}_S \otimes \ket{0}_A=\sum_\sigma \psi_\sigma\ket{\sigma}_S\otimes\ket{0}_A.
\end{equation}
Subsequently, one can apply the QSVT~\cite{Gilyen:2018khw,Martyn_2021} using the block-encoding $W$ and a chosen polynomial $P(x)$ of degree $m$. By selecting a sufficiently high degree m, we effectively sharpen the posterior distribution. Explicitly, this transformation maps the amplitudes as $\psi_\sigma\mapsto P(\psi_\sigma)\approx\psi_\sigma^m$.
Consequently, the new sampling probabilities scale approximately as $\propto(h(\sigma))^{2m}$ for amplitude encoding, and $\propto(h(\sigma))^m$ for Born encoding.
This power-law scaling aggressively suppresses low-probability configurations while isolating and amplifying the peaks (the modes) of the distribution. Because the success probability of the QSVT operation---marked by $\ket{0}_A$---depends heavily on the specific landscape of the data distribution, we assume standard amplitude amplification techniques are employed to post-select the desired state. Providing rigorous analytical bounds on this success probability is generally impossible without strong assumptions regarding the underlying data-generating process, and thus remains beyond the scope of this work.

\section{Discussion}\label{sec:discussion}
This work introduced a quantum implementation of a specific probabilistic model over permutations, based on spectral methods on the symmetric group. While the approach outlines a potential path towards a classically hard, quantumly feasible \textit{and} practically useful quantum machine learning algorithm, it constitutes only a first step in this direction. 

Further work is needed to estimate and optimise the circuit's constant-prefactor resources to ensure they are feasible under realistic quantum hardware constraints. For example, the gate complexity of the $\sn$-QFT relies heavily on asymptotic estimates (see the Appendix in \cite{Larocca:2024waq,Bravyi:2023veg}), which offer only a crude measure of the actual physical overhead required. Additionally, we have not yet provided empirical evidence comparing the model's performance against classical bandlimited implementations or other non-spectral heuristics on real-world datasets. This remains a challenging task that requires a careful benchmark design, as super-exponential runtimes limit us to even smaller scales than the usual exponential cost of simulations. 

Benchmarks will be crucial to identify the right balance between feasibility on early-stage fault-tolerant hardware, classically intractable problem sizes, and practical usefulness. We expect that this balance might be achieved at relatively modest scales of $n=15$ to $n=30$, which is well below the regime where the asymptotic limitations on the number of Markov chain steps $T$ become predictive of the success of the algorithm.

The motivation of the current study, however, was aimed at a deeper level: using a concrete example, we aimed to understand if the remarkable ability of quantum computers to efficiently move between direct and Fourier space can open up avenues for probabilistic modeling over group-structured data. We will therefore conclude with some more general insights from this endeavour.

\paragraph*{Learning as an application for the $\sn$-QFT.} Machine learning seems to be a yet unexplored, but promising area of application for non-Abelian Quantum Fourier Transforms. In some sense, the QFT over the symmetric group was responsible for the first major disappointment of early quantum algorithms research. 
QFTs were originally invented to solve Abelian Hidden Subgroup Problems such as Shor's famous algorithm \cite{shor1994algorithms, kitaev1995quantum}. 
After the discovery of an efficient algorithm for the $\sn$ QFT, hopes were high to solve an important non-Abelian Hidden Subgroup Problem with practical relevance: \textit{graph isomorphism}, or to decide if two sets of nodes and edges relate to the same graph. However, evidence that this was not possible quickly mounted (see references in \cite{childs2005quantum}), and Hidden Subgroup Problems subsequently came out of fashion. Recently, some authors have rediscovered the topic \cite{Larocca:2024waq,Bravyi:2023veg}, but again for rather abstract use cases: to compute the multiplicities of irreducible representations under certain assumptions on their dimensionality. The work of Persi Diaconis \cite{diaconis1988group}, and later of Risi Kondor and Jonathan Huang, provides convincing arguments to believe that the $\sn$-QFT could find its real purpose in statistical analysis and machine learning. Our work here can be understood as encouraging evidence that quantum algorithms for such applications are, at least in general, possible.  

\paragraph*{Quantum computers can enable spectral methods.} Spectral methods combined with quantum generative models are a promising area where quantum computers could have a real impact on machine learning. Quantum states are an elegant framework to manipulate a generative model in direct and Fourier space. As known from the literature of the \textit{spectral bias} in deep learning \cite{rahaman2019spectral, xu2019frequency}, the Fourier spectrum of a model contains crucial features for regularisation and model design. From the work of Diaconis we know that this generalises to the group Fourier spectrum, which is exactly what enables the diffusion and conditioning step in the Markov chain model investigated here. However, care has to be taken as we manipulate the amplitudes of the quantum state, not the sampling distribution that constitutes the generative model via the Born rule. How information is encoded therefore plays an important role. In addition to this, manipulating the amplitudes of quantum states is no easy feat; while modern quantum computing research has added many important routines such as block-encoding and the QSVT, highly non-unitary manipulations remain costly---here it limits the number of steps of the Markov chain. We found it interesting that sensible conditions for learning, for instance, that the likelihood in the Bayesian update is consistent with the prior model, also make the quantum algorithm feasible. 

\paragraph*{Finding the most likely data.} Quantum states can be naturally viewed as \textit{implicit} generative models \cite{rudolph2024trainability,Recio-Armengol:2025dpz}, i.e., machine learning models that can generate samples $x \sim p(x)$ without being able to compute probabilities $p(x)$ directly. This is an attractive use of a quantum computer, as every measurement result is a valuable data point. But a possibly even more interesting use case for such quantum models is inference with a flavour of optimisation, which is a very hard task: find a data point that has maximal, or at least very high, probability according to the model. Such a task could be solved by polynomially amplifying large amplitudes and suppressing small ones, a strategy with a growing toolbox of tricks in quantum computing research (see for example \cite{jordan2025optimization}).

Overall, we conclude that building on Diaconis' insights regarding the power of non-commutative harmonic analysis for statistics over permutation-structured data offers a promising avenue towards finding classically difficult, quantumly feasible \textit{and} practically useful applications for quantum computers. It also serves as a compelling example of how quantum computers can unlock spectral methods for machine learning. The path towards validating this claim, however, is full of challenges, stemming both from the highly technical underlying theory and from the root of the classical hardness itself: the super-exponential growth of the size of the symmetric group. This study, we hope, represents a first step on this path.

\section*{Acknowledgements}
We thank Nathan Wiebe for important discussions that helped us complete this work.

\bibliography{main}

\newpage
\appendix
\onecolumngrid
\newpage

\section{A primer in group representation theory for finite groups}\label{app:representation_theory}

    This section serves as a primer on representation theory for finite groups, to be then specialized to $\sn$, as a finite, non-commutative group. Excellent sources for an in-depth discussion are~\cite{terras1999fourier, diaconis1988group, serre1977linear}.
    
    \noindent We first introduce the mathematical notion of group as:
    \begin{definition}[Group]
        A group $G$ is a non-empty set equipped with a binary operation $\cdot$, closed under its action: $G\cdot G \rightarrow G$. Thus $\forall g_1, g_2 \in G$, $g_1 \cdot g_2 \equiv g_1 g_2\in G$, and the following properties hold:
        \begin{itemize}
            \item \textbf{Identity element:} $\exists\ e \in G $ such that $e\cdot g = g$ $\forall g \in G$;
            \item \textbf{Associativity:} $(g_1\cdot g_2)\cdot g_3 = g_1\cdot (g_2 \cdot g_3)$ $\forall g_1, g_2, g_3 \in G$;
            \item \textbf{Inverse element:} $\forall g\in G,\ \exists g^{-1}\ |\ g\cdot g^{-1}=g^{-1}\cdot g = e$. The element $g^{-1}$ is called the inverse of $g$. 
        \end{itemize}
    \end{definition}
    \noindent If the underlying set is finite, the group itself is called \emph{finite}, as opposed to \emph{infinite groups}. The number of the group elements is called group \emph{cardinality} or \emph{order}, and denoted $|G|$. In this work the focus is on the symmetric --or permutation-- group $\sn$, which is finite with $|\sn|=n!$.
    \begin{definition}[Group representation]
        Let G be a group, and V a vector space over some field $\mathbb{F}$; a (linear) representation of G on V is a group homomorphism:
        \begin{equation}
            \rho:G\rightarrow GL(V)
        \end{equation}
        where $GL(V)$ is the group of invertible linear maps $V\rightarrow V$.
    \end{definition}
        This formal definition introduces $\rho$ as a map from the group to a space of linear operators such that:
        \begin{equation}
            \rho(e) = \mathbb{I}
        \end{equation}
        where $\mathbb{I}$ is the identity operator on $V$, and:
        \begin{equation}
            \rho(g_1)\rho(g_2)=\rho(g_1g_2)
        \end{equation}
        for every $g_1,g_2\in G$. In other words, representations preserve the algebraic structure of the group. The dimension of the underlying vector space $V$ is also the dimension of the representation, denoted $d_{\rho}$. It is often convenient to work with \emph{matrix representations}, i.e. to associate to each group element a matrix $\rho(g)\in \mathbb{F}^{d_{\rho}\times d_{\rho}}$. An important result is expressed by the following:
        \begin{theorem}
            Every complex representation of a finite group is equivalent to a unitary representation.
        \end{theorem}

        The connection of this formalism with the finite groups spectral theory can be introduced through the distinction between \emph{reducible} and \emph{irreducible} representations.
        \begin{definition}[Subrepresentation] Let $\rho:G\rightarrow GL(V)$ and $\alpha:G\rightarrow GL(W)$ be two representations of G, with $V,W$ vector spaces such that W is a subspace of V. If $\alpha(g)W\subset W$, or identically if $\left.\alpha(g)\right|_W=\rho(g)$ ,$\ \forall g\in G$ then $\alpha(g)$ is called a subrepresentation of $\rho(g)$.
        \end{definition}
        A matrix interpretation of this is that there exists a basis of $V$ that produces the block matrix:
        \begin{equation}
            \rho(g)\equiv\left(\begin{array}{cc}
                \alpha(g) & \star \\
                0 & \star
            \end{array}\right)
        \end{equation}
        \begin{definition}[Irreducible representation]
            A representation $\rho$ is said to be irreducible is its only subrepresentations are the trivial ones: $\rho$ itself and the null one.
        \end{definition}
        The following important result \cite{serre1977linear} connects reducible and irreducible representations, acting as a bridge towards the spectral theory of finite groups, where the latter have a central role:
        \begin{theorem}
        \label{th_decomp}
            Every representation of a finite group G is a direct sum of irreducible representations.
        \end{theorem}
        When a generic reducible representation $\rho$ is expressed as the direct sum of its irreducible representations it is said to be \emph{completely reducible}. In a suitable basis, the matrices of the general $\rho$ are block diagonal, where each block is a sub-matrix associated to an irreducible representation.

        For the purposes of this work in Section \ref{sec:qft_operator} we have introduced the \emph{regular representation} which we here label $\rho_{\text{reg}}$, that is completely reducible into the direct sum of every possible irreducible representation (\emph{irrep}) of G.
        The regular representation $\rho_{\text{reg}}$ is $|G|$-dimensional; for the symmetric group this means it is $n!$-dimensional. To introduce it, it is customary to consider the $n!$ group elements of $\sn$ as forming an orthonormal basis $\left\{e_g: g \in G\right\}$ for the vector space of the group algebra $\mathbb{C}[G]$, with $\left\langle e_x, e_y\right\rangle=\delta_{x, y}$.
        \begin{definition}
            The group algebra $\mathbb{F}[G]$ for a field $\mathbb{F}$ and a group G is the set of all the linear combinations of a finite number of elements:
            \begin{equation}
                \sum_{g \in G} a_g e_g,
            \end{equation}
            with coefficients $a_g\in\mathbb{F}$, and basis elements $e_g$\footnote{Sometimes notation can be confusing, and the basis elements of the algebra $e_g$ can be indicated simply as $g$.}, with a multiplication operation, canonically defined on the basis elements by extending the group law: $e_x e_y=e_{x y}$.
        \end{definition}
        Generally the choice for $\mathbb{F}$ is the complex field $\mathbb{C}$. Then
        \begin{equation}
            \left[\rho_{\text {reg }}(x)\right]_{z, y}=\left\{\begin{array}{ll}
            1 & \text { if } z=x y, \\
            0 & \text { otherwise },
            \end{array} \quad x, y, z \in G\right.
        \end{equation}
        so that, at the level of basis vectors: $\rho_{\text {reg }}(x) e_y=e_{x y}$.
        The regular representation matrices are therefore of dimension $n!\times n!$ for every group element of the symmetric group. This is classically generally prohibitive to handle even for small $n$.
        
        By theorem \ref{th_decomp}$, \rho_{\text{reg}}$ can be decomposed as:
        \begin{equation}
            \rho_{\mathrm{reg}} \cong \bigoplus_{\lambda} d_\lambda \rho_\lambda
        \end{equation}
        where each possible irrep $\rho_\lambda$ appears with a multiplicity $d_\lambda$ equal to its dimension.
        One of the most important results in group representation theory is the \emph{Peter-Weyl theorem}, which states that the matrix entries of the irreducible representations belonging to the regular representation (which includes all of them) form an orthogonal basis for the class $L(G)$ of all the complex functions on the group G. In the case of finite groups this space can be identified with the group algebra $\mathbb{C}[G]$, and the distinction between $L(G)$ and $\mathbb{C}[G]$ is not necessary.
        The Fourier transform can then be defined as the map that block diagonalizes the regular representation, and hence projects a function from the algebra $\mathbb{C}[G]$ onto the irreducible representations of the group.
        Formally:
        \begin{definition}[Group Fourier transform]
            The \emph{group Fourier transform} of the function $f\in\mathbb{C}[G]$ on the irrep $\rho$ is
            \begin{equation}
            \label{gft}
                \widehat{f}(\rho)=\sum_{x \in G} f(x) \rho(x) \quad \rho \in \mathcal{R}
            \end{equation}
            where $\mathcal{R}$ is a complete set of inequivalent irreducible representations of G.
        \end{definition}
        This definition is completely general for finite groups, both abelian and non-abelian. It can be shown that for abelian groups the GFT exactly reconnects with the results of the spectral theory commonly introduced for periodic functions or functions over $\mathbb{R}$ and $\mathbb{Z}$. Instead, for non-abelian finite groups the generalized relations are often much more convoluted, as the Fourier coefficients $\hat{f}(\rho)$ are matrices rather than scalars.

        The following, fundamental result can be proved for finite groups:
\begin{theorem}[Convolution theorem]
                Consider $f_1$, $f_2$ elements of the group algebra $\mathbb{C}[\sn]$. Then:
                \begin{equation}
                    \widehat{[f_1 \star f_2]}(\rho) = \widehat{f_1}(\rho)\,\cdot\widehat{f_2}(\rho)
                \end{equation}
            \end{theorem}
            Hence the group Fourier transform for a finite group turns convolutions in input space, into an irrep-wise matrix product of the Fourier transforms of the starting functions.

            A further specialization of this result happens for \emph{class functions}, which can be introduced from the following definitions from group theory:
            \begin{definition}[Conjugation]
            Elements $g,h\in G$ are \emph{conjugate}, written $g\sim h$, if there exists $x\in G$ with
                \begin{equation}
                    h=xgx^{-1}.
                \end{equation}
            \end{definition}
            \begin{definition}[Conjugacy class]\label{app:def:conjugacy_class}
            For $g\in G$, the \emph{conjugacy class} of $g$ is
                \begin{equation}
                    \operatorname{Cl}(g)=\{\,xgx^{-1}\;:\;x\in G\,\}.
                \end{equation}
            \end{definition}
In general, it can be proved that the number of conjugacy classes of a group is equal to the number of its irreps.

            Finally, a class function can be defined as a function constant on a conjugacy class:
            \begin{definition}[Class function]
            $f: G \rightarrow \mathbb{C}$ is a class function if:
                \begin{equation}
                    f\left(x g x^{-1}\right)=f(g) \text { for all } x, g \in G .
                \end{equation}        
            \end{definition}

            Introducing class functions is crucial to elegantly implement convolutions in Fourier space, because they adhere to the following important lemmas proved in \cite{terras1999fourier}, around which much of this chapter revolves, and whose origin can be traced back to the Schur's lemma\cite{diaconis1988group}:
            \begin{lemma}
            \label{app:eq:Schur_lemma} Let $G$ be a finite group.
                \ 
                \begin{itemize}
                    \item[1:] A function f defined on G is a class function if and only if:
                    \begin{equation}
                        f\star h = h\star f\ \ \ \ \ \ ,\forall h:G\rightarrow \mathbb{C}
                    \end{equation}
                    \item[2:] If f is a class function, then the convolution operator built from it and acting on any other function $h:G\rightarrow \mathbb{C}$ is not only block diagonal in Fourier space, but diagonal, and acting on each irrep $\rho_{\lambda}$ as $c_{\lambda}\mathbb{I}$, with $c_{\lambda}$ some constant depending on the irrep.
                \end{itemize}
            \end{lemma}
\section{Irreducible representations of $\sn$}\label{app:irrep_sn}

            In Appendix \ref{app:representation_theory} irreducible representations have been introduced as the elementary blocks in which reducible representations can be decomposed. This section formalizes the general structure of the irreducible representations specifically for $\sn$. The discussion starts from the concepts of partition, Young diagrams, and Young tableaux in the context of the permutation group.

            \subsubsection{Partitions and Young diagrams}

                A partition of a positive integer $n$ is a tuple of positive integers:
                \begin{equation}
                    \lambda=\left(\lambda_1, \lambda_2, \ldots, \lambda_k\right)
                \end{equation}
                such that:
                \begin{equation}
                \label{partition_rule}
                    \lambda_1 \geq \lambda_2 \geq \cdots \geq \lambda_k>0 \quad \text { and } \quad \sum_{i=1}^k \lambda_i=n .
                \end{equation}

                The standard notation to denote that $\lambda$ is a partition of $n$ is: $\lambda \vdash n$. A useful graphical tool to depict partitions are Young diagrams, which are box diagrams where each row $i$ is composed by $\lambda_i$ boxes. 
                \newline
                Consider for example $S_3$; for $n=3$ the partitions are $(3)$, $(2,1)$, and $(1,1,1)$, associated to the Young diagrams:
                
                \begin{center}
                    \Yboxdim{1.2em}
                        (3):\quad \yng(3) \qquad (2,1):\quad \yng(2,1) \qquad (1,1,1):\quad \yng(1,1,1)
                    \end{center}

                Similarly, for $S_4$ the partitions can be obtained from the ones of $S_3$, and adding a box in every possible arrangement that doesn't violate Eq.\eqref{partition_rule}: $(4)$, $(3,1)$, $(2,2)$, $(2,1,1)$, $(1,1,1,1)$.

            \subsubsection*{Dimension of the irreducible representations}
            
                The irreducible representations of $\sn$ are in one-to-one correspondence with partitions $\lambda \vdash n$. The irrep associated with the partition $\lambda$ is in fact labeled $\rho_\lambda$, and its dimension $d_\lambda = \dim \rho_\lambda$ can be described combinatorially in terms of \emph{standard Young tableaux} of shape $\lambda$.

                \begin{definition}[Standard Young tableau]
                    Let $\lambda \vdash n$ and let $Y(\lambda)$ be its Young diagram. A \emph{standard Young tableau} of shape $\lambda$ is a filling of the boxes of $Y(\lambda)$ with the integers $1,2,\ldots,n$, such that the entries strictly increase along each row and each column.
                \end{definition}
                
                A fundamental result is that the dimension $d_\lambda$ of the irrep $\rho_\lambda$ is equal to the number of standard Young tableaux of shape $\lambda$.
                The usual simple example of $S_3$ can be used to illustrate this:
                \begin{itemize}
                    \item for $\lambda=(3)$ there is only one possible way to fill the Young diagram:
                        \young(123)
                        and for this reason $d_{(3)}=1$.
                    \item $\lambda=(2,1)$ has two possible tableaux:
                        \young(12,3) \qquad \young(13,2)
                        and in fact $d_{(2,1)}=2$.
                    \item $\lambda=(1,1,1)$ has one single possible filling: 
                        \young(1,2,3)
                        and $d_{(1,1,1)}=1$.
                \end{itemize}

                These dimensions satisfy the identity:
                \begin{equation}
                    \sum_{\lambda \vdash 3} d_\lambda^2=1^2+2^2+1^2=6=\left|S_3\right|
                \end{equation}
                which holds true for any $n$.

                It's important to understand that this simple case is very limited, and that the number of partitions, and tableaux explode rapidly with increasing $n$, making once again the problem of working with the full regular representation of $\sn$ complex to handle by a classical algorithm.
                There is however a closed form expression to compute the dimension of an irrep instead of building manually the tableaux; the equation is usually introduced as \emph{Hook-length formula}.

                For each box $c$ in the Young diagram $Y(\lambda)$, the \emph{hook length} $h(c)$ is defined as:

                \begin{align*}
                h(c)
                  &= \#\{\text{boxes to the right of } c \text{ in the same row}\} \\
                  \quad&+ \#\{\text{boxes below } c \text{ in the same column}\}
                    + 1.
                \end{align*}

                \begin{theorem}[Hook-length formula]
                \label{Hook}
                Let $\lambda \vdash n$ be a partition of $n$, and $Y(\lambda)$ its Young diagram. The dimension $d_\lambda$ of the irreducible representation $\rho_\lambda$ of $\sn$ is given by:
                
                \begin{equation}
                d_\lambda=\frac{n!}{\prod_{c \in Y(\lambda)} h(c)}.
                \label{app:eq:hook_length}
                \end{equation}
                \end{theorem}
                
                Therefore, each partition $\lambda \vdash n$ labels a unique irreducible representation $\rho_\lambda$ of $\sn$ with dimension $d_\lambda$, given either by the number of standard Young tableaux of shape $\lambda$ or equivalently by the Hook-length formula.
                
                Every finite-dimensional representation of $\sn$ can be decomposed into a direct sum of these irreps. In particular, the regular representation decomposes as
                \begin{equation}
                    \rho_{\mathrm{reg}}
                    \cong
                    \bigoplus_{\lambda \vdash n} d_\lambda\, \rho_\lambda,
                \end{equation}
                where each irrep $\rho_\lambda$ appears with multiplicity equal to its dimension $d_\lambda$.
\newpage
\section{Statistical interpretations of the Fourier spectrum} \label{app:spectrum_stat_interp}

The Fourier coefficients of a function $h:\mathbb{S}_n\rightarrow\mathbb{C}$,
\begin{equation}
    \label{fcoef}
    \left[\hat{h}_\lambda\right]_{i j}=\sum_\sigma h(\sigma) \rho_\lambda(\sigma)_{i j}\,,
\end{equation}
arise from expanding $h$ onto the matrix-element basis functions $\rho_\lambda(\cdot)_{ij}$ of each irreducible representation $\rho_\lambda:\mathbb{S}_n\rightarrow GL(V_\lambda)$.

When $h$ is a probability distribution over permutations, these coefficients acquire a statistical meaning: they encode the marginal probabilities of $h$ at various orders. To make this precise, we introduce the \emph{marginal subspaces} $M_\lambda$ (Young permutation modules), which collect all observable statistics of a given order---e.g., $M^{(n-1,1)}$ contains all first-order marginals. A basis for $M_\lambda$ is provided by the \emph{indicator functions} $m_\lambda(\sigma)_{i,j}$, each of which equals $1$ when a specific assignment of objects to positions holds under~$\sigma$ and $0$ otherwise. For example, a first-order indicator function tests whether \textit{object $i$ is in position $j$}, while a second-order one tests whether \textit{objects $i_1, i_2$ are in positions $j_1, j_2$} simultaneously.

To formally describe how these marginal subspaces decompose into irreducible representations (Specht modules), we rely on a foundational result known as Young's Rule. A specific irreducible representation $V_\mu$ may appear multiple times within the decomposition of a marginal subspace $M_\lambda$. This multiplicity is quantified by the Kostka number, denoted $K_{\mu\lambda}$. Combinatorially, $K_{\mu\lambda}$ counts the number of semistandard Young tableaux of shape $\mu$ and weight $\lambda$. Algebraically, it tells us exactly how many copies of the ``pure'' frequency component $V_\mu$ exist inside the observable marginal space $M_\lambda$.

With this definition in hand, Young's Rule expresses the decomposition of $M_\lambda$ as a direct sum over all integer partitions $\mu$ that dominate $\lambda$:
\begin{equation}
    M_\lambda \cong \bigoplus_{\mu \trianglerighteq \lambda} K_{\mu \lambda}\, V_\mu\,.
\end{equation}
If $h$ is a probability distribution, evaluating its Fourier transform at the indicator functions yields the exact marginal probabilities:
\begin{equation}
    \hat{m}(\lambda)_{i,j}=\sum_\sigma h(\sigma)\, m_\lambda(\sigma)_{i,j}\,.
\end{equation}
A full derivation of this correspondence can be found in \cite{JMLR:v10:huang09a}. Below, we illustrate the key ideas through the first non-trivial case.

\paragraph{First-order marginals and the $(n-1,1)$ irrep.}
The first-order marginal probability $p_i(j)$ is the probability that object $i$ is mapped to position $j$, independently of the remaining $n-1$ objects:
\begin{equation}
    p_i(j) = \sum_{\sigma : \sigma(i)=j} h(\sigma)\,.
\end{equation}
The $n \times n$ matrix $P$ with entries $P_{i,j} = p_i(j)$ fully characterizes the first-order statistics of $h$.

By Young's Rule, the first-order marginal subspace decomposes as $M^{(n-1,1)} \cong V^{(n)} \oplus V^{(n-1,1)}$, where $V^{(n)}$ is the trivial (constant, uniform) component and $V^{(n-1,1)}$ is the standard irrep. Since a uniform distribution assigns probability $1/n$ to every object--position pair, the $V^{(n-1,1)}$ component of the marginals is obtained by subtracting this baseline:
\begin{equation}
    \left[\hat{h}_{(n-1,1)}\right]_{i,j} = p_i(j) - \frac{1}{n}\,.
\end{equation}
These centered marginals correspond directly to the Fourier coefficient $\hat{h}_{(n-1,1)}$.

\paragraph{Dimensional consistency.}
This correspondence can be verified by counting the degrees of freedom of the marginal matrix $P$. Normalization of $h$ requires each row to sum to $1$:
\begin{equation}
    \sum_{j=1}^n p_i(j) = \sum_{j=1}^n \sum_{\sigma: \sigma(i)=j} h(\sigma) = \sum_{\sigma \in \mathbb{S}_n} h(\sigma) = 1\,,
\end{equation}
imposing $n$ row constraints. Similarly, every position must be filled by exactly one object, giving $n$ column constraints:
\begin{equation}
    \sum_{i=1}^n p_i(j) = 1\,.
\end{equation}
Since the total sum $\sum_{i,j} p_i(j) = n$ is implied by either set, one constraint is redundant, leaving $2n-1$ independent constraints. The number of free parameters in $P$ is therefore:
\begin{equation}
    \# \text{d.o.f.} = n^2 - (2n - 1) = (n-1)^2\,,
\end{equation}
which matches the size of the Fourier coefficient matrix $\hat{h}_{(n-1,1)}$, since $\dim V^{(n-1,1)} = n-1$ yields an $(n-1) \times (n-1)$ matrix.

\section{Details on the Diffusion operation}\label{app:diffusion}
\subsection{Non-Abelian convolution}
For the main text we have dropped the double time indices, and have kept only the time index for the distribution $h$. Here, we explicitly write them for clarity, and express the general case where the diffusion kernel $q$ is time dependent. Using the Markov chain assumption, we show that a diffusion operation on $h$ is equivalent to a non-abelian convolution: 
\begin{align}
    h(\sigma^{(t+1)}) &= \sum_{\sigma^{(t)}}h(\sigma^{(t+1)}|\sigma^{(t)})\cdot h(\sigma^{(t)}) \\
    &=\sum_{\{(\sigma^{(t)},\pi^{(t)}):\;\sigma^{(t+1)}=\pi^{(t)}\cdot\sigma^{(t)}\}} q^{(t)}(\pi^{(t)})\cdot h^{(t)}(\sigma^{(t)})\\
    &=\sum_{\sigma^{(t)}}q^{(t)}(\sigma^{(t+1)}\cdot (\sigma^{(t)})^{-1})\cdot h(\sigma^{(t)}) \label{app:eq:non_abelian_conv_diffusion_intermediate} \\
    &\equiv\left(q^{(t)}\star\, h \right)(\sigma^{(t+1)}).
    \label{app:eq:non_abelian_conv_diffusion}
\end{align}

\subsection{Markovian Matrix Formulation}\label{app:sec:markov_matrix_form}
Here we formulate the convolution in Eq.~\eqref{app:eq:non_abelian_conv_diffusion} as a Markov process, adopting the usual linear algebraic perspective. Let the state space be the finite group $G$ (here the symmetric group $\sn$) with cardinality $N = |G| = n!$. We fix an arbitrary indexing of the group elements $G = \{g_1, g_2, \dots, g_N\}$.
The probability distribution $h^{(t)}: G \to \mathbb{R}$ is represented as a column vector $\mathbf{h}^{(t)} \in \mathbb{R}^N$, where the $i$-th component corresponds to the probability mass on the $i$-th group element:
\begin{equation}
    (\mathbf{h}^{(t)})_i = h^{(t)}(g_i).
\end{equation}
We seek to express the time evolution as a linear transformation via a transition matrix $\mathbf{Q}^{(t)} \in \mathbb{R}^{N \times N}$, such that:
\begin{equation}
    \mathbf{h}^{(t+1)} = \mathbf{Q}^{(t)} \mathbf{h}^{(t)}.
\end{equation}
Writing this matrix-vector multiplication in index notation yields:
\begin{equation}
    (\mathbf{h}^{(t+1)})_i = \sum_{j=1}^{N} \mathbf{Q}^{(t)}_{ij} (\mathbf{h}^{(t)})_j.
    \label{eq:matrix_def}
\end{equation}
We compare Eq.~\eqref{eq:matrix_def} directly to the non-Abelian convolution derived in Eq.~\eqref{app:eq:non_abelian_conv_diffusion_intermediate}. Identifying the target state $\sigma^{(t+1)}$ with the group element $g_i$ and the source state $\sigma^{(t)}$ with $g_j$, Eq.~\eqref{app:eq:non_abelian_conv_diffusion_intermediate} becomes:
\begin{equation}
    h^{(t+1)}(g_i) = \sum_{j=1}^N q^{(t)}\left(g_i \cdot g_j^{-1}\right) \cdot h^{(t)}(g_j).
\end{equation}
By matching terms with Eq.~\eqref{eq:matrix_def}, we identify the entries of the transition matrix as:
\begin{equation}
    \mathbf{Q}^{(t)}_{ij} = q^{(t)}\left(g_i \cdot g_j^{-1}\right).
\end{equation}

\paragraph*{$G$-Circulant Structure.}
The matrix $\mathbf{Q}^{(t)}$ exhibits a specific symmetry: its entries depend solely on the ``difference" (via group inverse) between the row and column indices. Such a matrix is defined as a \textit{$G$-circulant matrix}. This is a direct generalization of standard circulant matrices used in classical signal processing, which arise when the underlying group is the cyclic group $C_n$ (Abelian). 
While standard circulant matrices are diagonalized by the Discrete Fourier Transform (DFT), $G$-circulant matrices are block-diagonalized by the Generalized Fourier Transform over the group $G$ (cf.~ Equation~\eqref{eq:group_fourier_transform}).
\paragraph*{Connection to Representation Theory.}
This matrix $\mathbf{Q}^{(t)}$ is formally the image of the diffusion kernel under the left regular representation of the group algebra. If we consider the kernel as an element $\mathcal{K} = \sum_{g \in G} q^{(t)}(g) g$ in the group algebra $\mathbb{R}[G]$, then $\mathbf{Q}^{(t)}$ is the matrix representing left multiplication by $\mathcal{K}$.
\paragraph*{Dimensions.}
The dimensionality of this system is determined by the size of the group. Since $G = \sn$, the state vector $\mathbf{h}^{(t)}$ has dimension $N = n!$. Consequently, the transition matrix $\mathbf{Q}^{(t)}$ has dimensions $n! \times n!$. For example, even for a small permutation group like $S_5$, the transition matrix is $120 \times 120$; for $S_{52}$, the dimension $52!$ renders explicit matrix construction computationally intractable, necessitating the algebraic convolution form derived in (E4).

\subsection{Scalability of the block-encoding}
In general, to block-encode a matrix $A$ one needs to rescale by $A/\alpha$ to ensure that the resulting matrix,
\begin{equation}
    D = \begin{pmatrix}
    A/\alpha & * \\
    * & * 
    \end{pmatrix},
    \label{eq:diffusion_operator_uniform_fourier_space}
\end{equation}
is unitary. One can choose $\alpha=\max |\sigma(A)|$, where $\sigma(A)$ are the singular values of $A$---$\sigma$ elsewhere in the text is a permutation---which in the case of block-encoding the diffusion operator $D$: $\alpha = \max |c_\lambda|$. 
For computing the character of transpositions we can assume that $\rho_\lambda$ is unitary without loss of generality due to the cyclicity of the trace. Furthermore, we use that unitary matrices have eigenvalues of modulus one.
\begin{equation}
    \chi_\lambda = \Tr\rho_\lambda = \sum_{\ell=1}^{d_\lambda} e_\ell \leq \sum_{\ell=1}^{d_\lambda} |e_\lambda| = d_\lambda.
\end{equation}
We can show in the same way $\chi_\lambda\geq -d_\lambda$, where  equality is achieved only for the alternating representation $\lambda=(1,1,\dots,1)$ when evaluating the character on transpositions. Therefore, $\abs{\chi_\lambda/d_\lambda}\leq 1$. Using $p\in [0,1]$:
\begin{equation}
    1\leq2p-1\leq c_\lambda \leq p+(1-p) = 1 \Rightarrow \abs{c_\lambda}\leq 1.
    \label{eq:c_lambda_bounds}
\end{equation}
Therefore, we can set the scaling factor $\alpha=1$. \\

\subsubsection{Efficient implementation of the block-encoding oracle}\label{app:block_encoding_oracle_implementation}

\paragraph*{Analytical formulas for $c_\lambda$.}
The eigenvalue $c_\lambda = p + (1-p)\,r_\lambda$ (Eq.~\eqref{eq:c_lambdas}) is determined by the normalised character ratio $r_\lambda\coloneqq\chi_\lambda((12))/d_\lambda$. For a partition $\lambda = (\lambda_1,\lambda_2,\dots,\lambda_\ell)$ with conjugate partition $\lambda'=(\lambda'_1,\lambda'_2,\dots)$ (obtained by transposing the Young diagram of $\lambda$, i.e., $\lambda'_i$ counts the boxes in the $i$-th column), the closed-form expression from \cite{diaconis1988group} and \cite{re1950some} reads:
\begin{equation}
    r(\lambda)=\frac{1}{n(n-1)} \sum_{i=1}^{\ell}\left[\lambda_i^2-(2 i-1) \lambda_i\right]=\frac{1}{\binom{n}{2}} \left(\sum_{i}\binom{\lambda_i}{2}-\sum_{i}\binom{\lambda_i^{\prime}}{2}\right),
\end{equation}
a special case for transpositions of a more general result from Frobenius.
For a single $\lambda$, computing $r_\lambda$ and consequently $c_\lambda$ is classically efficient. However, the diffusion must be applied to all frequencies, requiring these quantities for all partitions $\lambda \vdash n$. The number of partitions $p(n)$ grows superpolynomially, following the asymptotic expression\footnote{The formula was obtained by G. H. Hardy and Ramanujan in 1918 and independently by J. V. Uspensky in 1920.}
\begin{equation}
    p(n) \sim \frac{1}{4n\sqrt{3}}e^{\pi\sqrt{\frac{2n}{3}}}.
    \label{app:eq:ramanujan}
\end{equation}
While asymptotically $p(n)$ scales super-polynomially, for moderate $n$ where applications of interest might reside, it is much smaller than $n!$ (see Figure~\ref{fig:partitions_of_n_vs_n}).

\paragraph*{Quantum Read-Only Memory (QROM).}
When classical pre-computation of all $c_\lambda$ is feasible, one can load them via QROM. This approach is efficient for moderate $n\approx 30$--$50$, since $d_\lambda$ and $\chi_\lambda((2,1^{n-2}))$ are classically manageable for $n\lesssim 100$. However, the gate cost scales as $\mathcal{O}(p(n))$, which becomes infeasible for large $n$.

\paragraph*{Quantum Arithmetic.}
The combinatorial formulas above admit reversible arithmetic implementations using $\mathcal{O}(n\log n)$ qubits and $\mathrm{poly}(n)$ gates, to coherently compute $c_\lambda$ up to precision $\epsilon=2^{-b}$ (e.g., for $b=30$, $\epsilon\approx 10^{-9}$).
Hence, the diagonal operator $\hat{q}$ admits a block-encoding that is efficient in the standard asymptotic sense, requiring $\mathrm{poly}(n)$ qubits, $\mathrm{poly}(n)$ gates, and $\mathrm{polylog}(1/\varepsilon)$ precision overhead. There is a crossover point beyond which quantum arithmetic becomes preferable to QROM.

\paragraph*{Amplitude Encoding.}
The value $c_\lambda$ is encoded into the amplitude via a controlled rotation $R_y(\theta_\lambda)$, where $\theta_\lambda = 2\arccos(c_\lambda)$. 
In practice, the arithmetic or QROM step loads the binary approximation of the \textit{angle} $\theta_\lambda$ into the $k$-qubit register $\ket{\theta_\lambda}_c = \ket{b_1 b_2 \dots b_k}$. 
The rotation is then implemented as a sequence of fixed rotations controlled by the individual bits of the register:
\begin{equation}
    CR_y(\theta_\lambda) = \prod_{j=1}^k CR_y\left( \frac{\pi}{2^{j-1}} \right)^{(j)},
\end{equation}
where the $j$-th gate rotates the target ancilla by $\pi/2^{j-1}$ only if the $j$-th qubit of the register is $\ket{1}$.
The resulting state is:
\begin{equation}
    \ket{\lambda}\ket{\theta_\lambda}_c \left( \cos(\frac{\theta_\lambda}{2}) \ket{0}_A + \sin(\frac{\theta_\lambda}{2})\ket{1}_A \right).
\end{equation}
Finally, we uncompute the angle register. Conditioned on the ancilla being $\ket{0}_A$, we have successfully applied the amplitude $c_\lambda$.

\subsection{Success probability at $t=0$}\label{app:ps_proof_t0}
We first investigate the initial diffusion step at $t=0$. 
In the following, we show that we can exactly compute $p_s^{(0)}$, allowing us to used standard (Grover) amplitude amplification requiring $\mathcal{O}(1/\sqrt{p})$ iterations of amplitude amplification to amplify the success probability $p_s^{(0)}$ to nearly 1. For this step to be efficient we require $p\geq 1/\text{poly}(n)$.
\begin{align}  p_s^{(0)}&=\norm{\hat{q}\ket{\hat{\psi}^{(0)}}}^2\\
    &=\sum_{\lambda\vdash n}\sum_{i=1}^{d_\lambda}\abs{c_\lambda}^2\frac{d_\lambda}{n!}\\
    &=\sum_{\lambda\vdash n}\abs{c_\lambda}^2\frac{d_\lambda^2}{n!}
\end{align}
This sum is a weighted average of $|c_\lambda|^2$ under the \textit{Plancherel measure} $\mu_\lambda(n)=\frac{d_\lambda^2}{n!}$. 
The Plancherel measure concentrates on partitions far from the trivial representation.
Simultaneously, $|c_\lambda|^2$ decreases with dominance, i.e., increasing ``frequency''.
Hence, we see that the sum will decay with $n$, but the question is how fast, such that we can efficiently amplify the success probability. First, we compute for the case of uniform diffusion. 
Expanding $|c_\lambda|^2$
\begin{align}
    p_s^{(0)} &=\sum_{\lambda\vdash n}\left[p^2\frac{d_\lambda^2}{n!}+2p(1-p)\frac{d_\lambda\chi_\lambda((2,1^{n-2}))}{n!}+(1-p)^2\frac{\chi_\lambda((2,1^{n-2}))^2}{n!}\right].
    \label{eq:p_s_class_function_noise}
\end{align}
To evaluate each term, we use fundamental results in representation theory. The first term is equal to $p^2$, since 
\begin{equation}
\sum\limits_{\lambda\vdash n}d_\lambda^2=|G|=|\sn|=n!.
\end{equation}
For the second term, we observe that the sum over $d_\lambda\chi_\lambda((2,1^{n-2}))$ is actually the character of the regular representation evaluated on the element $(2,1^{n-2})$. This holds the regular representation decomposes to all irreps of $\sn$ each appearing in the block decomposition with multiplicity equal to their dimension. 
\begin{equation}
    \rho_{\text{reg}}=\bigoplus_{\lambda\vdash n}d_\lambda \rho_\lambda \Rightarrow \chi_{\text{reg}}(g) = \sum_{\lambda\vdash n}d_\lambda \chi_\lambda(g)
\end{equation}
Furthermore, the character of the regular representation satisfies
\begin{equation}
    \chi_{\text{reg}}(g) =\begin{cases}
        n!\quad \text{if }g=e\\
        0\quad\text{otherwise}.
    \end{cases}
\end{equation}
Hence, for $g=(2,1^{n-2})$ the second term in eq.~\eqref{eq:p_s_class_function_noise} is zero.

For the third term, we use the \textit{character orthogonality relations} and the \textit{Orbit-Stabilizer Theorem}, which we explain in the following.
The space of complex-valued \textit{class functions} of a finite group $G$ is equipped with an inner product:
\begin{equation}
    \Braket{\alpha,\beta}\coloneqq \frac{1}{|G|}\sum_{g\in G}\alpha(g) \overline{\beta(g)}
\end{equation}
The irreducible characters form an orthonormal basis in the space of class functions, themselves being class functions by definition due to the cyclicity of the trace.
This yields the so-called orthogonality relation between the rows of the character table:
\begin{equation}
\Braket{\chi_\mu,\chi_\nu} = \delta_{\mu\nu},
\end{equation}
where $\mu,\nu\in\Lambda_G$.
Furthermore, for $g,h\in G$, we have the orthogonality relation of the columns of the character table
\begin{equation}
\sum_{\chi_\lambda} \chi_\lambda(g)\overline{\chi_\lambda(h)}= \begin{cases}
    |C_G(g)|, \;\text{if $g$, $h$ are conjugate}\\
    0\qquad\text{otherwise}.
\end{cases}
\label{eq:character_ortho_relation_columns}
\end{equation}
where the sum is over all irreducible characters of $G$ and $C_G(g)$ the \textit{centralizer} of $g$.
We provide the following definitions for completeness. 

Let $G$ be a group and $S$ a subset, that is not necessarily a subgroup.

\begin{definition}[Centralizer]
The centralizer of $S$, denoted $C_G(S)$, is the set of elements of $G$ that commute with all elements of $S$:
\begin{equation}
    C_G(S)\coloneqq \{g\in G\;|\;gsg^{-1}=s\quad\forall s\in S\}.
\end{equation}
\end{definition}
It can be shown that $C_G(S)\leq G$.
\begin{definition}[Center of a group]
The subgroup $C_G(G)$ is the set of elements that commute with all elements of $G$, is denoted $Z(G)$, and is called the center of $G$. 
\end{definition}
The center of a group is by construction abelian.
\begin{definition}[Normalizer]
The normalizer of $S$ in $G$ is the set
\begin{equation}
    N_G(S)\coloneqq \{g\in G\;|\; gSg^{-1}=S\},
\end{equation}
where $gSg^{-1} =\{gsg^{-1} \;|\; s\in S\}$.
\end{definition}
In plain words the Normalizer is the set of all elements in $G$ that leave $S$ invariant. The normalizer is also a subgroup of $G$. It is weaker to be in the normalizer of a set than to be in its centralizer $C_G(S)\leq N_G(S)$. Both the centralizer and normalizer subgroups are determined by $G$ acting on subsets of itself by conjugation. We could also define another set, the \textit{stabilizer}, with respect to the group action rather than conjugation.
\begin{definition}[Stabilizer]
Let $s\in S$. The stabilizer of $s$ in $G$ is the set
\begin{equation}
    G_s = \{g\in G\;|\; g\cdot s = s\}.
\end{equation}
\end{definition}
Again, $G_s\leq G$. Note that there are many notations used in the literature for denoting the stabilizer, such as $\text{Stab}_G(s)$.

It can be shown that conjugacy is an equivalence relation and therefore partitions the group $G$ into equivalence classes, called conjugacy classes (cf.~Definition~\ref{app:def:conjugacy_class}).
Conjugacy classes have properties that are analogous to cosets; we informally state some of them here. Since conjugacy classes form equivalence classes they ``tile'' the group. Each element $g$ belongs to precisely one conjugacy class, and two conjugacy classes $\text{cl}(a)$ and $\text{cl}(b)$ are equivalent if and only if $a$ and $b$ are conjugate, and disjoint otherwise. The \textit{class number} of $G$ is the number of distinct (nonequivalent) conjugacy classes. All elements belonging to the same conjugacy class have the same \textit{order}. For the symmetric group, conjugacy classes can be categorized by the \textit{cycle type}, e.g., transpositions are 2-cycles and are order 2. 
\begin{proposition}
    Let a permutation $\sigma\in\sn$ be of cycle type $(\alpha_1,\alpha_2,\dots,\alpha_d)$, where $1\leq\alpha_1\leq\dots\alpha_i\leq\alpha_{i+1}\leq\dots\alpha_d\leq n$ is the length of each cycle and $\sum_{i=1}^d \alpha_i=n$. For example, $\tau=(132)(45)(67)$ is of cycle type $(3,2,2)$ and $\xi=(123)(456)(7)$ is of cycle type $(3,3,1)$. Fixed points, i.e., objects that are not permuted, such as $(7)$ in the previous example, can be omitted from the notation. Each cycle can be decomposed to $\alpha_i-1$ transpositions. Subsequently, the number of total transpositions needed to express $\sigma$ is
    \begin{equation}
        \sum_{i=1}^{d}(\alpha_i-1)=n-d
    \end{equation}
\end{proposition}
\begin{definition}[Orbit]
    When $G$ acts on a set $S$ we call an orbit of $s\in S$, denoted $\text{Orb}_G(s)$, the set of all elements in $S$ to which $s$ can be moved to by the action of elements in $G$:
    \begin{equation}
        \text{Orb}_G(s) = \{g\cdot s\;|\; g\in G\}
    \end{equation}
\end{definition}
\begin{theorem}[Orbit-Stabilizer Theorem]
    For any element $s\in S$ the size of its orbit multiplied by the size of its stabilizer is equal to the order of the group:
    \begin{equation}
        |\text{Orb}_G(s)| \cdot |\text{Stab}_G(s)|=|G|
    \end{equation}
\end{theorem}

The properties of the objects defined above have many more interesting interplays that are outside the scope of this work.
A particular interplay that is useful for computing the desired sum emerges in the context where one considers the group action as a \textit{conjugation} where a group $G$ acts on itself. One can show that for $g,h\in G$ $ghg^{-1}$ indeed fulfills the axioms of group action. 
In this context, we can apply the Orbit-Stabilizer theorem to the conjugacy class and centralizer of an element $s$, where we identify the former as the orbit of $s$ and the latter as the stabilizer of $s$:
\begin{equation}
    |\text{cl}(s)| =\frac{|G|}{|C_G(s)|}.
\end{equation}
For the symmetric group and for the conjugacy class of transpositions which is of size $\binom{n}{2}$:
\begin{equation}
    |C_{\sn}((2,1^{n-2}))| = \frac{n!}{\binom{n}{2}}= 2\cdot(n-2)!
\end{equation}
Applying the above result together with the orthogonality relation in eq.~\eqref{eq:character_ortho_relation_columns}, evaluated at $g=h=(2,1^{n-2})$, we compute the sum appearing in the third term of eq.~\eqref{eq:p_s_class_function_noise} to be $\frac{2(1-p)^2}{n(n-1)}$.

Putting everything together for the probability of success $p_s$ for a diffusion that is a class function we have
\begin{equation}
    p_s = p^2+\frac{2(1-p)^2}{n(n-1)}.
\end{equation}
This means that we can efficiently implement amplitude amplification in $\mathcal{O}(1/\sqrt{p_s})=\mathcal{O}(n)$ steps.

\subsection{Success Probability at an arbitrary time step}\label{app:ps_proof_arbitrary}
The (unnormalized) target state after $d$ consecutive diffusion steps is:
\begin{equation}
    \hat{q}^d\ket{\hat{\psi}^{(t)}} = \frac{1}{\sqrt{N^{(t)}}}\sum_{\lambda\vdash n}\sum_{i,j=1}^{d_\lambda} c_\lambda ^d \left[\hat{h}^{(t)}_\lambda\right]_{ij} \ket{\lambda ij}.
\end{equation}
The success probability is the squared norm of this unnormalized state:
\begin{align}
    p_s^{(t)} &= \norm{\hat{q}^d\ket{\hat{\psi}^{(t)}}}^2 \nonumber \\
    &= \frac{1}{N^{(t)}} \sum_{\lambda\vdash n} |c_\lambda|^{2d} \norm{\hat{h}^{(t)}_\lambda}_{\text{HS}}^2 \label{eq:ps_fourier_sum} \\
    &= \frac{1}{\norm{h_\lambda^{(t)}}_2^2} \hat{q}^d \cdot \norm{\hat{h}^{(t)}_\lambda }_2^2 \quad \text{(Parseval's Identity)} \\
    &= \frac{\| q^{\star\, d} * h^{(t)} \|^2_2}{\norm{ h^{(t)}}^2_2} \\
    &= \sum_{\lambda\vdash n}|c_\lambda|^{2d}\frac{d_\lambda}{n!}\sum_{\sigma,\xi \in \sn} \frac{1}{N^{(t)}}\chi_\lambda(\xi^{-1}\sigma) h^{(t)}(\sigma) h^{(t)}(\xi)
\end{align}
where $\|\cdot\|_{\text{HS}}$ denotes the Hilbert-Schmidt norm $\Tr(A^\dagger A)$, and we used the property that the Fourier transform of a convolution is the point-wise product of Fourier coefficients.

\subsubsection{Lower Bound for Lazy Random Walks}
We first consider the case where the random walk is ``lazy,'' meaning the probability of remaining at the identity $p>1/2$.

\begin{theorem}[Lazy Walk Bound]
    Let $h: \sn \to \mathbb{C}$ be a function on the symmetric group representing the distribution at step $t$. Let $q$ be the transition kernel of a random walk defined by self-loop probability $p \in (0.5, 1]$. Then, for any $n \ge 2$, and $d\geq 1$ consequtive diffusion steps, the diffusion operation success probability satisfies:
    \begin{equation}
        p_s^{(t)} \ge (2p - 1)^{2d}.
    \end{equation}
    This bound is independent of $n$ and the support size of $h$.
\end{theorem}

\begin{proof}
    Starting from Eq.~\eqref{eq:ps_fourier_sum}, we can view $p_s^{(t)}$ as a weighted average of $|c_\lambda|^{2d}$:
    \begin{equation}
        p_s^{(t)} = \frac{\sum_{\lambda \vdash n} |c_\lambda|^{2d}  \|\hat{h}_\lambda\|_{\text{HS}}^2}{\sum_{\lambda \vdash n} \|\hat{h}_\lambda\|_{\text{HS}}^2}.
    \end{equation}
    Since norms are non-negative, the sum is lower-bounded by the smallest eigenvalue:
    \begin{equation}
        p_s^{(t)} \ge \min_{\lambda \vdash n} |c_\lambda|^{2d} .
    \end{equation}
    The eigenvalues of the diffusion kernel are given by $c_\lambda = p + (1-p) \frac{\chi_\lambda((2,1^{n-2}))}{d_\lambda}$. The character ratio $\frac{\chi_\lambda((2,1^{n-2}))}{d_\lambda}$ for a transposition lies in the interval $[-1, 1]$. The minimum value $-1$ is attained by the alternating (sign) representation $\lambda=(1^n)$.
    Therefore, for $p > 1/2$:
    \begin{equation}
        \min_\lambda |c_\lambda| = \left| p + (1-p)(-1) \right| = |2p - 1|.
    \end{equation}
    Taking the power $2d$ yields the result.
\end{proof}

\subsubsection{Lower Bound for General Rational Walks}
We now extend this result to general random walks where $p$ is any rational number, covering the regime where $p \le 1/2$.

\begin{theorem}[General Rational Walk Bound]
    Let $p = \frac{a}{b} \in (0, 1)$ be a rational number ($a,b\in\mathbb{Z}^+$). Assume $p$ is generic such that $c_\lambda \neq 0$ for all $\lambda \vdash n$. Then the success probability satisfies the polynomial lower bound:
    \begin{equation}
        p_s^{(t)} \ge \left(\frac{4}{b^2 n^4}\right)^d.
    \end{equation}
\end{theorem}

\begin{proof}
    As before, $p_s^{(t)} \ge \min_{\lambda} |c_\lambda|^{2d} $. Substituting $p = a/b$ into the eigenvalue equation:
    \begin{equation}
        c_\lambda = \frac{a}{b} + \left( \frac{b-a}{b} \right) \frac{\chi_\lambda((2,1^{n-2}))}{d_\lambda}.
    \end{equation}
    Using the Frobenius formula, the normalized character ratio for the transposition class is determined by the contents of the Young diagram $\lambda$:
    \begin{equation}
        \frac{\chi_\lambda((2,1^{n-2}))}{d_\lambda} = \frac{1}{\binom{n}{2}} \sum_{(i,j) \in \lambda} (j-i) = \frac{s_\lambda}{r},
    \end{equation}
    where $s_\lambda \in \mathbb{Z}$ is the sum of contents and $r = n(n-1)/2$. Substituting this back:
    \begin{equation}
        c_\lambda = \frac{a r + (b-a)s_\lambda}{b r}.
    \end{equation}
    The numerator $K = a r + (b-a)s_\lambda$ is an integer linear combination of integers, so $K \in \mathbb{Z}$. By the assumption $c_\lambda \neq 0$, we must have $K \neq 0$, implying $|K| \ge 1$. Thus:
    \begin{equation}
        |c_\lambda| = \frac{|K|}{b r} \ge \frac{1}{b \frac{n(n-1)}{2}} > \frac{2}{b n^2}.
    \end{equation}
    Taking the power $2d$ yields the result.
\end{proof}

\paragraph*{Remark on pathological parameters.}
    Establishing a polynomial lower bound for \textit{all} real parameters $p \in (0,1)$ faces a subtle theoretical obstruction due to the existence of Liouville numbers. These ``pathological'' irrationals can be approximated by rationals with super-polynomial accuracy, potentially causing the integer linear combination in our derivation (and thus the spectral gap) to vanish arbitrarily fast. 
    However, Liouville numbers constitute a set of measure zero. Furthermore, in any computational realization, $p$ is represented by finite-precision floating-point numbers, which are inherently rational. Thus, for all practical purposes, the polynomial lower bound $p_s^{(t)} \in \Omega(n^{-4})$ holds, guaranteeing algorithmic feasibility.

\subsubsection{Proofs for the case of Born encoding}
In Sec.~\ref{sec:born_encoding}, we introduced the Born encoding quantum model: 
\begin{align}
   \ket{\psi^{(t)}} &= \sum_{\sigma\in\sn}\psi_{\sigma}^{(t)}\ket{\sigma}\\
   &=\sum_{\sigma\in\sn}\sqrt{h^{(t)}(\sigma)}\ket{\sigma}.
\end{align}
The above proofs for the lower bounds of the success probability remain exactly the same, as they depend solely on the spectrum of the diffusion operator, not on the specific encoding of the probability distribution.
Starting from the definition of success probability as the squared norm of the unnormalized state after diffusion:
\begin{equation}
    p_s^{(t)} = \| \hat{q}^{\star \,d} |\hat{\psi}^{(t)}\rangle \|^2 = \sum_{\lambda \vdash n} |c_\lambda|^{2d} \| \hat{\psi}_\lambda^{(t)} \|_F^2.
\end{equation}
Note that the state $|\hat{\psi}^{(t)}\rangle$ is normalized, so $\sum_\lambda \| \hat{\psi}_\lambda^{(t)} \|_F^2 = 1$. Consequently, $p_s^{(t)}$ represents a convex combination (weighted average) of the squared eigenvalues $|c_\lambda|^{2d} $. This sum is strictly lower-bounded by the minimum eigenvalue:
\begin{equation}
    p_s^{(t)} \geq \min_{\lambda \vdash n} |c_\lambda|^{2d} .
\end{equation}
The remainder of the proof follows exactly as in the standard, amplitude encoding case. Thus, the lower bounds of $p_s^{(t)}$ are the same, regardless of whether the amplitudes encode probabilities $h$ or square-root probabilities $\sqrt{h}$.

\subsection{Diffusion on quantum amplitudes}
\label{app:sec:diffusion_amplitudes}
Here we provide some details for the generative model that uses Born encoding.
\subsubsection{Spectral Relationship: Generalized Auto-Convolution}
The relationship between the spectra of $h$ and $\psi = \sqrt{h}$ is determined by the point-wise product $h(\sigma) = \psi(\sigma) \cdot \psi(\sigma)$. In Fourier space, this product corresponds to a generalized convolution over the irreducible representations. Following the derivation for the product of two functions provided in App.~\ref{app:sec:conditioning_in_Fourier}, the Fourier component of $h$ at a target irrep $\nu$ is expressed as:

\begin{equation}
    [\widehat{h}]_\nu = \frac{1}{\sqrt{d_\nu} \sqrt{|\sn|}} \sum_{\lambda, \mu \in \hat{G}} \sqrt{d_\lambda d_\mu} \sum_{\ell=1}^{z_{\lambda \mu \nu}} \left[ C_{\lambda \mu}^\dagger \left( \hat{\psi}_\lambda \otimes \hat{\psi}_\mu \right) C_{\lambda \mu} \right]_{(\nu, \ell)}
    \label{eq:gen_auto_conv}
\end{equation}
up to a choise of the normalization (unitary or not),
where:
\begin{itemize}
    \item $\hat{\psi}_\lambda$ and $\hat{\psi}_\mu$ are the Fourier matrices of the amplitude $\psi$ at irreps $\lambda$ and $\mu$.
    \item $C_{\lambda \mu}$ is the unitary Clebsch-Gordan matrix that decomposes the tensor product representation into a direct sum of irreps:
    \begin{equation}
        C_{\lambda \mu}^\dagger \left( \rho_\lambda(\sigma) \otimes \rho_\mu(\sigma) \right) C_{\lambda \mu} = \bigoplus_{\nu} \bigoplus_{\ell=1}^{z_{\lambda \mu \nu}} \rho_\nu(\sigma).
    \end{equation}
    \item $z_{\lambda \mu \nu}$ is the multiplicity (or Clebsch-Gordan series coefficients also called Kronecker coefficients) indicating how many copies $\ell$ of irrep $\nu$ appear in the decomposition of $\lambda \otimes \mu$.
    \item The notation $[\cdot]_{(\nu, \ell)}$ indicates the extraction of the $d_\nu \times d_\nu$ sub-block corresponding to the $\ell$-th copy of irrep $\nu$ from the block-diagonalized matrix.
\end{itemize}
This formulation explicitly shows that the spectrum of the probability distribution $h$ is formed by mixing the spectral components of the amplitude $\psi$ through the Clebsch-Gordan coefficients of the group $G$.
In the specific case where $G$ is Abelian, the irreps are 1-dimensional ($d_\rho=1$) and the tensor product reduces to scalar multiplication with index addition ($\rho_i \otimes \rho_j \to \rho_{i+j}$), recovering the standard discrete convolution theorem: $\widehat{h}_k = \sum_{i} \widehat{\psi}_i \widehat{\psi}_{k-i}$.

\subsubsection{Connection to imaginary time evolution}
In Sec.~\ref{sec:model_use}, we proposed applying the diffusion kernel $q$ directly to the $\ket{\psi}$ (Born encoding setting) rather than the probability distribution $h$ (amplitude encoding setting). Here, we identify this operation as a step of \textit{Imaginary Time Evolution} (ITE) on $\sn$.
In quantum mechanics, the time evolution of a state $\ket{\psi}$ is governed by the Schrödinger equation:
\begin{equation}
    i\hbar \frac{\partial}{\partial t} \ket{\psi(t)} = \hat{H} \ket{\psi(t)},
\end{equation}
where $\hat{H}$ is the system Hamiltonian. 
Performing a Wick rotation by substituting real time with imaginary time $\tau = it$ transforms the oscillatory evolution into a relaxation process:
\begin{equation}
    -\hbar \frac{\partial}{\partial \tau} \ket{\psi(\tau)} = \hat{H} \ket{\psi(\tau)}.
    \label{eq:imaginary_time_se}
\end{equation}
The solution to Eq.~\eqref{eq:imaginary_time_se} is given by the decay operator $\ket{\psi(\tau)} = e^{-\frac{\tau}{\hbar}\hat{H}}\ket{\psi(0)}$. Unlike unitary real-time evolution, this non-unitary operator suppresses high-energy eigenstates, driving the system toward the ground state of $\hat{H}$ as $\tau \to \infty$.
In the context of our random walk on $\sn$ the natural Hamiltonian is the Laplacian operator $\Delta$ associated with the Cayley graph of the generated random walk. The Laplacian measures the ``roughness'' or kinetic energy of the distribution. For the lazy random transposition walk defined in Eq.~\eqref{eq:diffusion_kernel}, the Hamiltonian is given by $\hat{H} = \mathbb{I} - \hat{T}$, where $\hat{T}$ is the adjacency operator weighted by transition probabilities.
The diffusion equation on the group generated by this Hamiltonian is the heat equation:
\begin{equation}
    \frac{\partial \psi}{\partial \tau} = \Delta \psi.
\end{equation}
Fundamental analysis on groups establishes that the fundamental solution (Green's function) to the heat equation is the heat kernel $q_\tau$. Therefore, the time-evolved state after a discrete time step $\delta \tau$ is exactly the convolution of the initial state with the kernel:
\begin{equation}
    \psi(\sigma, \tau + \delta \tau) = (q_{\delta \tau} \star \psi(\cdot, \tau))(\sigma).
\end{equation}
Comparing this to our operational definition in Eq.~\eqref{eq:diffusion_step_fourier_space}, we see that applying the convolution $q \star \psi$ is mathematically equivalent to applying the propagator $e^{-\delta \tau \hat{H}}$. 
Consequently, our alternating sequence of diffusion ($q \star \psi$) and conditioning ($\sqrt{h(\varphi|\sigma)} \cdot \psi$) steps constitutes a Trotter-Suzuki decomposition of the evolution under a composite Hamiltonian $\mathcal{H}_{total} = \hat{H}_{kinetic} + \hat{H}_{potential}$, where the potential is defined by the negative log-likelihood of the data, $\hat{H}_{potential} \sim -\log L$. This provides a natural interpretation of our model as a variational procedure that drives the state toward the ground state of a data-defined Hamiltonian.

\section{Details on the Conditioning operation}\label{app:sec:conditioning}

\subsection{Encodings of permutations}
\label{app:sec:conditioning:lehmer}

In this work, we the permutations $\sigma\in\sn$ are encoded in two main ways. First, we have the \emph{Cauchy} or \emph{explicit} encoding $c(\sigma)$. Here we store explicitely the permutation vector associated to $\sigma$. More precisely $c(\sigma)_i$ (often referred to as $\sigma(i)$) stores the the new index in position $i$ after the permutation $\sigma$ is applied to the vector of indices $(1\;2\;\dots \;n)$, giving
\begin{equation}
\label{Cauchydef}
    c(\sigma)\equiv\left(c(\sigma)_1, c(\sigma)_2, \ldots, c(\sigma)_n\right)
\end{equation}
We are interested in this encoding as it is the natural language for observed data. Another important encoding is the so called \emph{Lehmer code} $\ell(\sigma)$, which can be similarly described as
\begin{equation}
\label{Lehmerdef}
    \mathcal{\ell}(\sigma)\equiv\left(\ell(\sigma)_1, \ell(\sigma)_2, \ldots, \ell(\sigma)_n\right)\,.
\end{equation}
Here, instead of explicitely storing the full permutation vector, we define $\ell(\sigma)$ entrywise as
\begin{equation}
\label{Lehmerdef_counting}
    \ell(\sigma)_i=\#\left\{j>i: c(\sigma)_j<c(\sigma)_i\right\}.
\end{equation}
where $\#$ denotes the cardinality of a set. This definition, can be more informally read as \textit{consider position i in the original explicit permutation vector and then count the number of occurrences for which $c(\sigma)_i>c(\sigma)_j$ when $j>i$}.
The following example illustrates how a Lehmer string is decoded:
\begin{example}[Lehmer decoding for $\mathbb{S}_4$]
\label{example:lehmer}
            Let us consider as an example $n=4$ and the Lehmer string {$\ell(\sigma)=(\ell(\sigma)_1,\ell(\sigma)_2,\ell(\sigma)_3,\ell(\sigma)_4)=(2,2,1,0)$} with $\ell(\sigma)_i\in\{0,\ldots,4-i\}$ and $\ell(\sigma)_4=0$.
            To decode $\ell(\sigma)$ into the one-line Cauchy encoding $c(\sigma)=(c(\sigma)_1,c(\sigma)_2,c(\sigma)_3,c(\sigma)_4)$ we start from a ordered list
            $A=[1,2,3,4]$ and, for $i=1,\dots,4$, select the $\ell(\sigma)_i$-th available element of $A$, set it as $c(\sigma)_i$, and
            remove it from $A$:
    
            \begin{center}
            \begin{tabular}{@{} c c c l @{}}
            \toprule
            $i$ & $\ell(\sigma)_i$ & $A$ & $c(\sigma)$ \\
            \midrule
            1 & 2 & $[1,2,3,4]$ & $c(\sigma)_1=3$, \quad $A\gets[1,2,4]$ \\
            2 & 2 & $[1,2,4]$   & $c(\sigma)_2=4$, \quad $A\gets[1,2]$   \\
            3 & 1 & $[1,2]$     & $c(\sigma)_3=2$, \quad $A\gets[1]$     \\
            4 & 0 & $[1]$       & $c(\sigma)_4=1$, \quad $A\gets[\ ]$    \\
            \bottomrule
            \end{tabular}
            \end{center}
            hence the decoded permutation is $\sigma=[3,4,2,1]$.
        \end{example}
There is a bijection between the Lehmer codes of Eq.\eqref{Lehmerdef} and the set of integers $\{ 0, 1,..., n!-1\}$, hence this encoding is denser than the natural, explicit encoding of permutations, which by definition admits no repetition in the digits.

We use this encoding as it is the natural choice for the QFT implementation. When encoding Lehmer codes as quantum states, the qubit registers reflect the digits in Eq.\eqref{Lehmerdef}. Notice that each $\ell(\sigma)_i$ has a decreasing range as $i$ grows, in fact:
\begin{equation}
    \ell_1 \in \{0,\ldots,n-1\},\quad
    \ell_2 \in \{0,\ldots,n-2\},\ \ldots,\ 
    \ell_{n-1} \in \{0,1\},\quad
    \ell_n = 0.
\end{equation}
Thus, the number of qubits required for the $i$-th digit is:
\begin{equation}
    \left\lceil \log_2 (n-i+1) \right\rceil,
\end{equation}
and the total number of qubits is:
\begin{equation}
\begin{aligned}
    n_{\text{qubits}}^{L}
  &= \sum_{i=1}^{n-1} \left\lceil \log_2 (n-i+1) \right\rceil
   = \sum_{k=2}^{n} \left\lceil \log_2 k \right\rceil\,,\\
\end{aligned}
\end{equation}
thus requiring at most $n_{\text{qubits}}^{L} \in \mathcal{O}(n \log n)$ to be implemented

\subsection{The reorder-update approach}
\label{app:sec:conditioning:ru}

While observations are generally provided in the explicit Cauchy encoding of permutations $c(\sigma)$, by construction, the Bayes update step needs to be performed in the Lehmer encoding $\ell(\sigma)$. Since as noted in Eq.~\eqref{Lehmerdef_counting}, given a permutation $\sigma$ each element $\ell(\sigma)_i$ is related to the value and relative position of \emph{all} entries $c(\sigma)_i$ of the explicit Cauchy encoding, selecting permutations consistent with observations generally requires a decoding procedure analogous to Example~\ref{example:lehmer}. While such a decoding could be implemented efficiently by quantum arithmetic using ancillary qubits, for specific kinds of observations this step can be greatly simplified by taking advantage of the structure of $\ell(\sigma)$.

As an example, note that the \emph{first}
element $\ell(\sigma)_1$ is uniquely determined by the value of $c(\sigma)_1$, namely $\ell(\sigma)_1 = c(\sigma)_1 -1$. Continuing, $\ell(\sigma)_2$ is uniquely determined by $c(\sigma)_1$, $c(\sigma)_2$, and in general
$\ell(\sigma)_i$ can be computed knowing just the first $i$ elements of the explicit permutation mapping. Since all permutations $\sigma$ sharing the first $i$ elements $c(\sigma)_i$ will result in the same first $i$ Lehmer code elements $\ell(\sigma)_i$, consistency with an order $k$ partial assignment can be easily checked in both encodings if the assigned indices are $1,2,\dots i$.

Similarly, the \emph{last} non-trivial element $\ell(\sigma)_{n-1}$ is uniquely determined by the ordering of $c(\sigma)_{n-1}$ and $c(\sigma)_n$, namely

\begin{equation}
    \ell(\sigma)_{n-1} = \begin{cases}
        1 &  c(\sigma)_{n-1} >c(\sigma)_n\\
        0 & \text{otherwise}
    \end{cases}
\end{equation}
Continuing the pattern, also $\ell(\sigma)_{n-2}$ can be uniquely determined by the ordering of $c(\sigma)_{n-2}$, $c(\sigma)_{n-1}$ and $c(\sigma)_n$, and in general $\ell(\sigma)_{n-1-i}$ can be computed knowing just the ordering of the last $i+1$ elements of the explicit permutation mapping.
Since all permutations $\sigma$ sharing the same ordering of the last $i+1$ elements $c(\sigma)_{n-i}$ will result in the same last $i$ Lehmer code elements, consistency with a partial ranking can also be easily checked in the case where the ranked entries are $n, n-1,\dots n-i$.

This leads to the general reorder-update strategy, which allows to check wether partial assignments or partial rankings are consistent with a given permutation $\sigma$. Given a set of indices $\boldsymbol{i} = (i_1,\dots i_k)$, we first compose $\sigma$ with a fixed permutation $\pi$, with the property of rearranging the indices to the beginning (or end respectively) of the permutation vector. This can be interpreted as a change in the canonical ordering of objects, which facilitate the necessary checks. Quantumly, this is implemented as a unitary operation $U_\pi$. The explicit procedure to select $\pi$ and implement $U_\pi$ to apply the permutation in the Lehmer code is provided in App.~\ref{app:sec:conditioning:ru:reorder}. As a second step we check for consistency, by analytically computing the first (or last respectively) digits of the Lehmer code of a conistent permutation, and comparing it to $\sigma$. This step is equivalent to the computation of $\varphi(\sigma)$ in the general treatment of Section \ref{sec:conditioning}. Details on such a computation and the quantum implementation of can be found in App.~\ref{app:sec:conditioning:ru:update}. Finally, we apply the the inverse permutation $\pi^{-1}$ (i.e. $U_\pi^\dagger$ on the quantum model), effectively restoring the canonical ordering, and resetting the system for subsequent calculations.

\subsubsection{Reordering step: implementation of $U_\pi$}
\label{app:sec:conditioning:ru:reorder}

Implementing permutation compositions in the Lehmer code, i.e. the transformation $\ell(\sigma)\to\ell(\sigma\pi)$, can be achieved by decomposing $\pi$ into elementary transpositions $\tau_k$. Indeed, we have the following Proposition.

\begin{proposition}[Update rule for elementary transpositions] \label{prop:update_rule}
Given an elementary transposition $\tau_k = (k\,k+1)$ and an arbitrary $\sigma\in\sn$, then $$\ell(\sigma\tau_k)_i = \ell(\sigma)_i\;\;\forall i \notin \{k,k+1\}.$$
Furthermore, when $i\in \{k,k+1\}$, we have the following update rules
\begin{equation*}
    \ell(\sigma\tau_k)_k = \begin{cases}
        \ell(\sigma)_{k+1} &\text{if}\;\; \ell(\sigma)_{k} > \ell(\sigma)_{k+1}\\
         \ell(\sigma)_{k+1}+1 &\text{otherwise}
    \end{cases}
\end{equation*}
and
\begin{equation*}
    \ell(\sigma\tau_k)_{k+1} = \begin{cases}
        \ell(\sigma)_{k} -1 &\text{if}\;\; \ell(\sigma)_{k} > \ell(\sigma)_{k+1}\\
         \ell(\sigma)_{k} &\text{otherwise}
    \end{cases} \\
\end{equation*}
\end{proposition}

\begin{proof}
By Eq.~\ref{Lehmerdef_counting}, it is clear that the relative position of $c(\sigma)_k$ and $c(\sigma)_{k+1}$ do not matter in the evaluation of $\ell(\sigma)_i\;\forall i \notin \{k,k+1\}$, as in those cases, either both or neither of them are included in the definition, and hence their contribution remain unchanged after applying $\tau_k$. 

Concerning the opposite case, consider as an example $\ell(\sigma\tau_k)_k$. After the elementary transposition is applied, $c(\sigma)_k$ and $c(\sigma)_{k+1}$ get swapped. For this reason, $\ell(\sigma\tau_k)_k$ is closely related to $\ell(\sigma)_{k+1}$, as they both refer to the same number $c(\sigma)_{k+1}$. However, now $c(\sigma)_{k+1}$ is one position behind, and as a result it 'sees' one more digit in Eq.~\ref{Lehmerdef_counting}, namely $c(\sigma)_{k}$. If $c(\sigma)_{k+1} > c(\sigma)_{k}$, its Lehmer code entry must be increased by one, and vice versa if $c(\sigma)_{k} > c(\sigma)_{k+1}$ it should remain the same. Since $c(\sigma)_{k}>c(\sigma)_{k+1} \Leftrightarrow \ell(\sigma)_{k} > \ell(\sigma)_{k+1}$, we have the first condition. The second one can be derived analogously.
\end{proof}

\begin{claim}
    Applying the update rule of Prop.~\ref{prop:update_rule} to the Lehmer encoding is  a unitary and local operation, i.e. it only involves qubits encoding $\ell(\sigma)_k$ and $\ell(\sigma)_{k+1}$.
\end{claim}

Indeed, the unitarity immediately follows from the fact that $\tau_k$ is a permutation and $\ell$ is a bijection, and therefore invertible. In this sense, such an update can be deterministically applied to the state $\ket{\ell(\pi)}$ acting locally and without needing ancillary systems. Let $U_k$ denote such a transformation. In general, implementing a reversible transformation without ancillary qubits requires no more than $\mathcal{O}(2^{n_{\text{bits}}}n_{\text{bits}})$ (Theorem 1, \cite{Wu24}). Since each entry of Lehmer code requires $n_{\text{bits}} <\log n$, we can construct $U_{k}$ with $\mathcal{O}(n\log n)$ depth. Deviating from the ancilla-free case, tighter bounds can be derived, e.g. linear depth in $n_\text{bits}$, by introducing ancillary systems and using quantum arithmetic operations.

Note that by composing multiple instances of $U_k$ for varying $k$, all permutations $\pi\in\sn$ can be applied to the Lehmer code of $\sigma$ by Prop. \ref{prop:update_rule}. In particular, we can build $U_\pi = \prod_kU_k$ such that
\begin{equation}
    \ket{\ell(\sigma)} \to U_\pi\ket{\ell(\sigma)} = \ket{\ell(\sigma\pi)}.
\end{equation}
This procedure has a cost proportional to the number of elementary transpositions $\pi$ can be decomposed into.

A particularly relevant class of permutations for what follows are $t$-cycles, i.e. those of the form
\begin{equation}
    \pi_t = (t\,1\,2\;\cdots t-1) = \prod_{k=t-1}^1 \tau_k\,,
\end{equation}
each of which requires $t\leq n$ elementary operations $U_k$ to be implemented. This is thanks to the following observation.

\begin{claim}\label{obs:reoder_beginning}
    Given a list of $k$ indices $\bm{i} = (i_1, \cdots i_k)$, we can map them to the list $(k,k-1,\cdots,1)$ using a combination of $\pi_t$, with a total cost scaling at most as $\sim kn$ elementary operations $U_k$.
\end{claim}

In other words, we can bring the indices contained in $\bm{i}$ to the beginning of the Lehmer code efficiently. Indeed this can be done following the procedure outlined in here. 

First, we sort the list $\bm{i}$ to put it in ascending order. In this way, we can assume that $i_1 < i_2<\cdots < i_k$. Then, we sequentially apply $\pi_{i_1}, \pi_{i_2},\cdots\pi_{i_k}$ to the Lehmer code using the respective $U_{\pi}$. Each operation brings the corresponding $i_l$ to the beginning of the line, and moves the subsequent right by one spot. Since the permutations become increasingly large as $l$ increases, we are assured that after each step, the position $i_{l'}$ of of subsequent indices remains unchagend, thus showing the correctness of the procedure. In particular, the full transformation can be summarized compactly as, 
\begin{equation}
    U_b = \prod_{l=1}^k U_{\pi_{i_l}}\,,
\end{equation}
which by definition of $U_\pi$ has the implementation cost stated above. We denote this transformation as $U_b$, as it reorders the elements of the set $\{1,\cdots,n\}$ in order to bring the indices $\bm{i}$ to the \emph{beginning} of the Lehmer code. The same strategy can also be applied to bring all indices to the \emph{end}.

\begin{claim}\label{obs:reorder_end}
 Given a list of $k$ indices $\bm{i} = (i_1, \cdots i_k)$, we can map them to the list $(n-k,n-k+1,\cdots,n)$ using a combination of $\tilde{\pi}_t$, with a total cost scaling at most as $\sim kn$ elementary operations $U_k$. 
\end{claim}

Indeed the procedure is very similar. First we sort the indices in ascending order, and then we sequentially apply $\tilde{\pi}_{i_k}, \tilde{\pi}_{i_{k-1}},\cdots\tilde{\pi}_{i_1}$, where each $\tilde{\sigma}_t$ is defined by
\begin{equation}
    \tilde{\pi}_t = (t+1\cdots\, n \, t ) = \prod_{k=t}^{n-1}\tau_k\,,
\end{equation}
using the respective $U_{\tilde{\pi}}$. The full transformation $U_e$ obtained in this manner has indeed the property of reordering the elements of the set $\{1,\cdots,n\}$ in order to bring the indices $\bm{i}$ to the \emph{end} of the Lehmer code.

\subsubsection{Update step: implementation of $\varphi$}
\label{app:sec:conditioning:ru:update}

Exploiting the properties of Eq.~\ref{Lehmerdef_counting}, it is always possible to analytically compute the first (or last) digits of Lehmer codes consistent with a partial assignment (or ranking). By comparing such result with the Lehmer code of a permutation $\sigma$, we effectively compute $\varphi(\sigma)$ of Section \ref{sec:conditioning}. Here, we provide the analytical solutions in both cases.

Concerning partial assignements of the form $c(\sigma)_{i_1}=j_1 \;\wedge \; c(\sigma)_{i_2}=j_2 \;\wedge \; \dots\; \wedge c(\sigma)_{i_k}=j_k$, and assuming without loss of generality that the indices $\boldsymbol{i}$ are stored in ascending order, this is immediately given by
\begin{equation}
    \ell(\sigma)_k = {c(\sigma)_{i_k}-\sum_{l<k}(c(\sigma)_{i_k}>c(\sigma)_{i_l})-1} = {j_k-\sum_{l<k}(j_k>j_l)-1}
\end{equation}
where the notation $(j_k>j_l)$ evaluates to 1 if the condition is satisfied and 0 otherwise. Since this does not depend on the specific permutation $\sigma$ but only on observed data, such calculation can be performed classically. In this way, the coherent calculation of $\varphi(\sigma)$ can be easily implemented by a multi-controlled NOT (MCX) gate, selecting only the correct, classically computed, bit representation of consistent $\ell(\sigma)_k$. Since we only act on the first $k$ entries of $\ell(\sigma)$, we effectively control the operation on $n_{\text{bits}}<k\log n$, which in the absence of ancillary qubits requires $\mathcal{O}(n_{\text{bits}}^2) = \mathcal{O}(k^2 \log^2 n)$ depth \cite{Barenco95}. 

Concerning partial rankings, after rearranging the indices $\boldsymbol{i}\to\bar{\boldsymbol{i}}$ to be in ascending order, it is easy to map the original condition $c(\sigma)_{i_1} > c(\sigma)_{i_2} \; \dots\; > c(\sigma)_{i_k} > c(\sigma)_{i_{k+1}}$, to an equivalent set of inequalities expressed in terms of the new $\bar{i}_k$. With this in mind, the last entries of consistent Lehmer codes are given by
\begin{equation}
   \ell(\sigma)_{n-l} =  \sum_{j=k}^{k-l+1}(c(\sigma)_{\bar{i}_j}>c(\sigma)_{\bar{i}_{k+1}})
\end{equation}
with the same notation as above. Again this can be efficiently computed classically based on the data, and a $\mathcal{O}(k^2 \log^2 n)$ deep circuit will suffice to implement $\varphi(\sigma)$ coherently.

\subsection{Success probability}

In this section, we compute the success probability of the conditioning step. To this end, we first introduce the quantity $h(\varphi)$ as the total proability of generating a permutation consisten with the data in the quantum model. Given the soft likelihood of Eq.~\ref{eq:def_likelihood}, we can express $h(\varphi)$ explicitely as
\begin{equation}
    h(\varphi) = s \Pr (\varphi(\sigma) = 1) + (1-s)(1-\Pr (\varphi(\sigma) = 1))
\end{equation}
with $\Pr (\varphi(\sigma) = 1) = \sum_{\sigma\, : \, \varphi(\sigma)=1} h(\sigma)$. With this in mind, the result is summarized by the following Proposition.

\begin{proposition}[Success probability]
Let $p_s$ denote the probability of successfully applying the conditioning operator $C(\varphi)$ defined in Section \ref{sec:conditioning}. Then we have
    \begin{equation}
        p_s \coloneqq \left\| c_\varphi \ket{\psi^{(t)}} \right\|^2 = h(\varphi)^2\frac{N^{(t+1)}}{N^{(t)}},
    \end{equation}
where $N^{(t)}$ denote the model normalization at step $t$.  
\end{proposition}

\begin{proof}

The success probability of the block-encoding of Section~\ref{sec:conditioning}, can be directly computed using the definition of $c_\varphi$ and of the quantum model $\ket{\psi^{(t)}}$. First we expand both definitions:
\begin{equation}
\begin{aligned}
    c_\varphi \ket{\psi^{(t)}} &= \frac{1}{\sqrt{N^{(t)}}}\sum_{\sigma \in \sn} h^{(t)}(\sigma)\, c_\varphi\ket{\sigma} = \frac{1}{\sqrt{N^{(t)}}}\sum_{\sigma \in \sn} h^{(t)}(\sigma) h(\varphi |\sigma)\,\ket{\sigma}\\
\end{aligned}
\end{equation}

We continue by using Bayes theorem $h^{(t)}(\sigma) h(\varphi |\sigma) = h(\varphi) h^{(t)}(\sigma|\varphi)$, where $h(\varphi)$ was already introduced. Since by definition, the step $t+1$ is obtained by incorporating data into the model, we have $h^{(t)}(\sigma |\varphi) = h^{(t+1)}(\sigma)$, yielding

\begin{equation}
\begin{aligned}
        c_\varphi \ket{\psi^{(t)}} &= \frac{1}{\sqrt{N^{(t)}}}\sum_{\sigma \in \sn}  h^{(t)}(\varphi) h^{(t)}(\sigma |\varphi)\,\ket{\sigma}
        = \frac{h^{(t)}(\varphi)}{\sqrt{N^{(t)}}}\sum_{\sigma \in \sn}   h^{(t+1)}(\sigma)\,\ket{\sigma}\\
        &= h^{(t)}(\varphi)\sqrt{\frac{{N^{(t+1)}}}{{N^{(t)}}}} \frac{1}{\sqrt{N^{(t)}}}\sum_{\sigma \in \sn}   h^{(t+1)}(\sigma)\,\ket{\sigma} = h^{(t)}(\varphi)\sqrt{\frac{{N^{(t+1)}}}{{N^{(t)}}}} \ket{\psi^{(t+1)}}\,.
\end{aligned}
\end{equation}

Computing the norm as in the definition of $p_s$, we get

\begin{equation}
    p_s = \left\|c_\varphi \ket{\psi^{(t)}}\right\|^2 =  \left(h^{(t)} (\varphi)\right)^2\frac{{N^{(t+1)}}}{{N^{(t)}}} \braket{\psi^{(t+1)}} = \left(h^{(t)} (\varphi)\right)^2\frac{{N^{(t+1)}}}{{N^{(t)}}}\,,
\end{equation}
which concludes the proof.
\end{proof}

\subsection{Conditioning in Fourier space}\label{app:sec:conditioning_in_Fourier}

In this section, we derive an expression for implementing conditioning (Bayes update) in Fourier space. While such an approach would reduce the cost in terms of the number of QFT calls involved, i.e., one would only need to implement QFT (t=0) and QFT$^\dagger$ ($t=T$), we show that implementing Bayes rule in Fourier space is equivalent to a hard problem; namely, computing Kronecker (Clebsch-Gordan) coefficients of $\sn$.

The element-wise Hadamard multiplication of the likelihood and prior distributions in input space is a composition that, by convolution theorem, becomes a clear convolution in Fourier domain only for abelian groups. For non-commutative groups, like $\sn$, it is possible to write a generalized expression in Fourier space, but it is much more complex, because different irreps are combined with each other. The following sub-sections present a compact review of the classical works by Risi Kondor \cite{kondor2011non, pmlr-v2-kondor07a} and Jonathan Huang \cite{JMLR:v10:huang09a, huang2007efficient}.

\subsubsection{Kronecker conditioning}
This derivation follows \cite{JMLR:v10:huang09a}, and presents the general framework to perform conditioning (Bayesian update) in Fourier domain.
The inverse group Fourier transform when working with real, orthogonal representations for which $\rho_\lambda\left(\sigma^{-1}\right)=\rho_\lambda(\sigma)^T$, can be defined as
\begin{equation}
\label{iqft}
    f(\sigma)=\frac{1}{|G|}\sum_\lambda d_{\rho_\lambda}Tr\left[\hat{f}_{\rho_\lambda}^T\cdot\rho_\lambda(\sigma)\right]
\end{equation}
Starting from:
    \begin{equation}
    \label{bayes}
        P(\sigma \mid z)=\eta \cdot P(z \mid \sigma) \cdot P(\sigma),
    \end{equation}
the goal is to write the product in Eq.\eqref{bayes} as:
\begin{equation}
\label{pro}
    f(\sigma)\cdot g(\sigma)=\frac{1}{|G|}\sum_\nu d_{\rho_\nu}\operatorname{Tr}\left(R_\nu^T\cdot\rho_\nu(\sigma)\right)
\end{equation}
where 
\begin{equation}
\label{product}
    \left[\widehat{fg}\right]_{\rho_\nu}=R_\nu
\end{equation}
is the Fourier transform of the pointwise product between the functions f and g on the irrep labeled by $\nu$.

The following are the main steps to derive a closed-form expression for (\ref{product}), the full derivation is in \cite{JMLR:v10:huang09a}.

First Eq.\eqref{iqft}) is applied to the LHS of Eq.\eqref{pro}:

\begin{equation}
\begin{aligned}
    f(\sigma)\cdot g(\sigma)&=\quad\left[\frac{1}{|G|}\sum_{\lambda}d_{\rho_{\lambda}}\operatorname{Tr}\left(\hat{f}_{\rho_{\lambda}}^{T}\cdot\rho_{\lambda}(\sigma)\right)\right]\cdot\left[\frac{1}{|G|}\sum_{\mu}d_{\rho_{\mu}}\operatorname{Tr}\left(\hat{g}_{\rho_{\mu}}^{T}\cdot\rho_{\mu}(\sigma)\right)\right]\\&=\quad\left(\frac{1}{|G|}\right)^2\sum_{\lambda,\mu}d_{\rho_\lambda}d_{\rho_\mu}\left[\mathrm{Tr}\left(\hat{f}_{\rho_\lambda}^T\cdot\rho_\lambda(\sigma)\right)\cdot\mathrm{Tr}\left(\hat{g}_{\rho_\mu}^T\cdot\rho_\mu(\sigma)\right)\right]
\end{aligned}
\end{equation}

It can be proved by properties of the tensor product that:
\begin{equation}
    \begin{aligned}
      \mathrm{Tr}\bigl(\hat f_\lambda^T\,\rho_\lambda(\sigma)\bigr)\,
      \mathrm{Tr}\bigl(\hat g_\mu^T\,\rho_\mu(\sigma)\bigr)
      &=
      \mathrm{Tr}\Bigl(\bigl(\hat f_\lambda^T\,\rho_\lambda(\sigma)\bigr)\otimes\bigl(\hat g_\mu^T\,\rho_\mu(\sigma)\bigr)\Bigr)\\[1ex]
      &=
      \mathrm{Tr}\Bigl(\bigl(\hat f_\lambda\otimes\hat g_\mu\bigr)^T\,\bigl(\rho_\lambda(\sigma)\otimes\rho_\mu(\sigma)\bigr)\Bigr)\,.
    \end{aligned}
\end{equation}   

Since $\rho_\mu$ and $\rho_\lambda$ are two irreps, their composition can be re-expressed in terms of irreps $\rho_\nu$ through a similarity transformation mediated by the C matrices, which are commonly called Clebsch-Gordan:
\begin{equation}
    C_{\lambda\mu}^{-1}\cdot\left[\rho_\lambda\otimes\rho_\mu\right](\sigma)\cdot C_{\lambda\mu}=\bigoplus_{\nu}\bigoplus_{\ell=1}^{z_{\lambda\mu\nu}}\rho_\nu(\sigma)
\end{equation}
here $z_{\lambda\mu\nu}$ is a coefficient that controls the number of copies of the irrep $\rho_\nu$.

Defining for convenience $A_{\lambda \mu} \triangleq C_{\lambda \mu}^{-1} \cdot\left(\hat{f}_{\rho_\lambda} \otimes \hat{g}_{\rho_\mu}\right) \cdot C_{\lambda \mu}$, it can be proved --the reader is referred to the original study-- that the Fourier transform of the element-wise product of two functions $f(\sigma)$, $g(\sigma)$ is:
\begin{equation}
\label{KC}
    [\widehat{f g}]_{\rho_\nu}=\frac{1}{d_{\rho_\nu}|G|} \sum_{\lambda \mu} d_{\rho_\lambda} d_{\rho_\mu} \sum_{\ell=1}^{z_{\lambda \mu \nu}} A_{\lambda \mu}^{(\nu, \ell)}
\end{equation}
Notice how, compared to a convolution in input space which becomes an irrep-wise multiplication in Fourier space, this expression involves different irreps, and combines them through Clebsch-Gordan coefficients.

\subsubsection{Conditioning through twisted Fast Fourier transform}

    This section is dedicated to a second possibility for the Bayes update step, as in \cite{pmlr-v2-kondor07a}. This method, called twisted fast Fourier transform (FFT) is deeply rooted in the hierarchical structure of the Fourier transform, and avoids the evaluation of the complex Eq.\ref{KC}.

    Kronecker conditioning is the underlying structure needed to perform in Fourier space, what in input space is a point-wise product. A simplification can be engineered for likelihoods which are built on indicator functions like $\mathds{1}_{i,j} \equiv$ \{\emph{object i belongs to track j}\}:
    \begin{equation}
    \label{likelihood}
    P\left(O_{i\to j}|\sigma\right)=
        \begin{cases}
            \pi&\mathrm{if}\ \sigma(i)=j,\\
            (1-\pi)/(n-1)&\mathrm{if}\ \sigma(i)\neq j.
        \end{cases}
    \end{equation}
    the difference with a true indicator function is that the image is not $\{1,0\}$, but softer. Another way to see this, closer to the twisted FFT routine is through the introduction of stabilizers.

    Let $H$ be a subset of $S_n$, then it is possible to define:
    \begin{equation}
    \label{stabilizer}
        H_i=\left\{h \in S_n: h(i)=i\right\},
    \end{equation}
    hence $H_i$ acts as a stabilizer for the $i$-th element of the set on which a generic group element $\sigma\in S_n$ acts. $H_i$ is isomorphic to $S_{n-1}$ because considers the permutations that affect every set element apart from the i-th.

    The link back to $S_n$ can be found defining a \emph{coset} of H as:
    \begin{equation}
        g_j H_i=\left\{g_j \circ h: h \in H_i\right\} \subset S_n
    \end{equation}
    where $g_j$ are the group elements in $S_n$ (there are many of them) bringing the i-th element in the j-th position. 

    In \cite{pmlr-v2-kondor07a} the following formalism is introduced for the cycle $[[j,n]]$:
    \begin{equation}
        J_{j,n}\equiv[[j,n]](i)= \begin{cases}i+1 & \text { if } j \leq i \leq n-1 \\ j & \text { if } i=n \\ i & \text { otherwise }\end{cases}
    \end{equation}
    which is fixing the elements from 1 to $j-1$, and performing a shift to the right by 1 for every element from index j to index n, in a cyclic fashion. This means that the n-th element is brought to the j-th position, much like $g_j$ is doing, but with a specialization of eq.(\ref{stabilizer}) with $i$ set to the element $n$.

    It is now possible to write:
    \begin{equation}
    \label{1cycle}
        \sigma = J_{j,n} \sigma^\prime \ \ \ \ \ \ \ \ \ \ \ \ \ \ \ \ \ \ \text{with}\ \sigma\in S_n\ , \sigma^\prime \in S_{n-1}
    \end{equation}
    where $\sigma^\prime$ is leaving the n-th element of the set untouched.

    Finally, from the definition of group Fourier transform:
    \begin{equation}
    \widehat{f}(\rho_\lambda)=\sum_{\sigma \in S_n} f(\sigma) \rho_\lambda(\sigma) \quad \rho_\lambda \in \mathcal{R}
    \end{equation}
    and by eq.(\ref{1cycle}) it is possible to write:
    \begin{equation}
    \label{expansion1}
        \widehat{f}(\rho_\lambda)=\sum_{j=1}^n \rho_\lambda([[j,n]])\sum_{\sigma^\prime\in S_{n-1}}\rho_\lambda(\sigma^\prime)f_j(\sigma^\prime)
    \end{equation}
    where, with a compact notation:
    \begin{equation}
        f_j(\sigma^\prime) = f([[j,n]]\sigma^\prime)
    \end{equation}
    From Eq.\eqref{expansion1} it is possible to recognize the group Fourier transform over the $S_{n-1}$, labeled by the partitions $\lambda^-$ of $n-1$ obtained with the constraint that only one box has been removed from the original partition $\lambda$ of $S_n$:
    \begin{equation}
    \label{generalGZdec}
        \widehat{f}(\rho_\lambda)=\sum_{j=1}^n \rho_\lambda([[j,n]])\bigoplus_{\lambda^-}\widehat{f_j}(\rho_{\lambda^-})
    \end{equation}
    This decomposition is exact working in the Gel’fand-Tsetlin (GT) basis, otherwise some change of basis matrices would enter the above equation .

    Specializing this formalism to first order marginals it is possible to introduce two-sided cosets to describe a process where the action on the set elements can be qualitatively described as: "bring the set element i to n, then permute all the elements according to some group element belonging to a stabilizer for element n, then move the untouched set element n to j". This is a possible way of associating element i and position j accessing the subgroup structure of $S_n$. In a more formal language:
    \begin{equation}
    \label{dcoset}
        f_{i\rightarrow j} = f([[j,n]]\sigma^\prime[[i,n]]^{-1})
    \end{equation}
    Plugging this into eq.(\ref{generalGZdec}) it is possible to write the final expression for the twisted FFT update: 
    \begin{equation}
    \label{kondorTFFT}
    \widehat{f}(\rho_\lambda)=\sum_{j=1}^n \rho_\lambda([[j, n]])\left[\bigoplus_{\lambda^{-}} \widehat{f}_{i \rightarrow j}\left(\rho_{\lambda^{-}}\right)\right] \rho_\lambda\left([[i, n]] ^{-1}\right)
    \end{equation}
    
    In particular, since $\rho_{(n-1)}$ is the trivial representation of $S_{n-1}$,
    the corresponding Fourier coefficient is exactly the first-order marginal:
    \begin{equation}
        \widehat{f}_{i\rightarrow j}\bigl(\rho_{(n-1)}\bigr)
        = P\bigl(\sigma(i)=j\bigr).
    \end{equation}
    which means that it is possible to access the marginal probability $P(\sigma(i)=j)$ exploiting the double coset structure and the hierarchy of Fourier transforms for subsets of $S_n$ (in the GT basis).
    \newline
    This approach is very useful when the likelihood function is written over cosets, as for the application to object tracking.
    From \cite{pmlr-v2-kondor07a} the scaling of this approximate approach to Bayes update is bounded to $\mathcal{O}(D^2n)$, where $D^2$ is the dimension of the largest irrep block within the band limit.

\section{From Plancherel's theorem to the QFT operator}
\label{app:sec:plancherel_qft}
The formalism of Appendix \ref{app:representation_theory} can be further enriched and specialized through the following result, which then leads to an important implication on the normalization of the group Fourier transform:

\begin{theorem}(Plancherel's theorem)
    Let $L(G)$ be the vector space of complex functions $f:G\to \mathbb{C}$ that is equipped with the inner product between $f_1,f_2\in L(G)$:
    \begin{equation}
        \Braket{f_1|f_2} \coloneqq \frac{1}{|G|}\sum_{g\in G}\overline{f_1(g)} f_2(g).
    \end{equation}
    Let $\hat{G}\coloneqq (\rho_1,\rho_2,\dots, \rho_m)$ be the set of non-equivalent irreducible representations of $G$. 
    For any $f\in L(G)$, its Fourier transform $\hat{f}(\rho)\in \hat{L}(G)$ at a irreducible representation $\rho$ of $G$ is a $d_\rho \times d_\rho$ matrix, where $d_\rho = \dim\rho$. The Hilbert space $\hat{L}(G)$ is the direct sum of matrix algebras corresponding to the irreps of $G$. If $V_1,V_2,\dots, V_m$ are the vector spaces for the irreps of $G$, correspondingly, then  
    \begin{equation}
    \hat{L}(G)=\bigoplus_{i=1}^{m} \text{End}(V_i),    
    \end{equation}
    where $\text{End}(V_i)$ is the space of linear maps from $V_i$ to itself, i.e., the space of endomorphisms, which can be represented as the space of $d_i\times d_i$ matrices. $\hat{L}(G)$ is equipped with the inner product:
        \begin{equation}
        \langle\hat{f}_1|\hat{f}_2\rangle \coloneqq \frac{1}{|G|} \sum_{\rho\in \hat{G}} d_\rho \Tr \left(\hat{f}_1(\rho)^\dagger \cdot \hat{f_2}(\rho)\right).
    \end{equation}
    The Fourier Transform is an isometry, in the sense that it maps $L(G)\mapsto \hat{L}(G)$ while preserving the inner products of the corresponding Hilbert spaces:
    \begin{equation}
        \Braket{f_1|f_2} = \langle\hat{f}_1|\hat{f}_2\rangle.
    \end{equation}
\end{theorem}
We call $L(G)$ direct space and $\hat{L}(G)$ Fourier space. Here, Plancherel's theorem is stated for finite groups but it also holds for locally compact (continuous) groups. The space of all functions $L(G)$ can be identified as the underlying vector space that corresponds to the regular representation of $G$ which the Fourier transform block-diagonalizes.  The above theorem is essentially the group-theoretic analogue of the Plancherel theorem in Fourier analysis, relating the ``energy" of a function to the ``energy" of its spectral components. Instead of ``energy'' conservation and distribution over frequencies, we can analogously think of ``probability''.

A special case of Plancherel's theorem for $f_1=f_2=f$ is the so-called \textit{Parseval's identity}, which relates the norms of the two spaces:
\begin{align}
    &\norm{f}_G^2 = \norm{\hat{f}}^2_{\hat{G}} \nonumber\Leftrightarrow \frac{1}{|G|}\sum_{g\in G}|f(g)|^2=\frac{1}{|G|}\sum_{\rho\in \hat{G}} d_\rho \Tr\left(\hat{f}(\rho)^\dagger\cdot\hat{f}(\rho)\right).
\end{align}
While an isometric map, the Fourier transform as defined in the theorem is not represented by a unitary operator, unless we use the unitary normalization discussed above. Then Parseval's identity becomes 
\begin{equation}
    \sum_{g\in G}|f(g)|^2=\sum_{\rho\in \hat{G}} \Tr\left(\hat{f}(\rho)^\dagger\cdot\hat{f}(\rho)\right)\equiv \sum_{\rho \in \hat{G}}||\hat{f}(\rho)||^2_\text{HS}
\end{equation}

In quantum computation the QFT operator is unitary:
\begin{equation}
    \mathcal{F} = \sum_\sigma\sum_{\lambda\vdash n}\sum_{i,j=0}^{d_\lambda-1}\sqrt{\frac{d_\lambda}{n!}}\left[\rho_\lambda(\sigma)\right]_{ij}\ket{\lambda ij}\bra{\sigma}.
\end{equation}
Therefore, to stay consistent with the unitarity of the QFT we normalize accordingly the FTs of $h$ and $q$ absorbing the necessary prefactors and the convolution in Fourier space takes the form of eq.~\eqref{eq:conv_fourier_space_unitary}. 
\paragraph*{Fourier spectrum implications.} When we choose unitary normalization, the Fourier coefficients are scaled by the dimension of a given irrep. This means that the Fourier coefficients $\hat{h}_{\rho_\lambda}$ that correspond to irreps of large dimensions---those with Young diagrams of roughly the same number of rows and columns---the dimensions of which grow very rapidly with $n$, are scaled by a large number $\sqrt{d_\lambda}$. \textit{Does that mean that statements about the spectrum such as low-frequency concentration are normalization dependent?} On the one hand, the absolute values are indeed scaled, inflated, and hence the spectrum is skewed towards irreps of high dimensionality, i.e. for $\lambda = (n-r, r)$ for $r\approx n/2$, which are far from the ``tails'' of the $\lambda$ partitions. On the other hand, when one uses the unitary normalization it is more natural to think of units of ``energy densities'' rather than the absolute values of the Fourier coefficients. In the unitary normalization (eq.~\eqref{eq:conv_fourier_space_unitary}), the convolution theorem ensures that the ``inflation'' factors $\sqrt{d_\lambda}$ cancel. As a result, the diffusion step simplifies to depend only on the intrinsic (normalized) Fourier spectrum of $q$.

\end{document}